\documentclass[journal]{IEEEtran}
\usepackage{amsmath,amssymb}
\usepackage{subfigure}
\usepackage{graphicx,graphics,color,psfrag}
\usepackage{cite,balance}
\usepackage{caption}
\allowdisplaybreaks
\usepackage{algorithm}
\usepackage{algorithmic}
\usepackage{accents}
\usepackage{amsthm}
\usepackage{bm}
\usepackage{url}
\usepackage[english]{babel}
\usepackage{multirow}
\usepackage{enumerate}
\usepackage{cases}
\usepackage{stfloats}
\usepackage{dsfont}
\usepackage{color,soul}
\usepackage{amsfonts}
\usepackage{cite,graphicx,amsmath,amssymb}
\usepackage{subfigure}
\usepackage{fancyhdr}
\usepackage{hhline}
\usepackage{graphicx,graphics}
\usepackage{array,color}
\usepackage{mathtools}
\usepackage{amsmath}

\newtheorem{definition}{\emph{\underline{Definition}}}

\newtheorem{lemma}{\emph{\underline{Lemma}}}

\newtheorem{proposition}{\emph{\underline{Proposition}}}

\newtheorem{example}{\bf Example}
\newtheorem{remark}{\bf \emph{\underline{Remark}}}

\def\({\left(}
\def\){\right)}

\def\b0{{\mathbf{0}}}

\newcommand{\nn}{\nonumber}

\begin{document}
\captionsetup[figure]{name={Fig.}}

\title{\huge SWIPT in Mixed Near- and Far-Field Channels: Joint Beam Scheduling and Power Allocation
} 
\author{Yunpu~Zhang,~\IEEEmembership{Student~Member,~IEEE}, and Changsheng~You,~\IEEEmembership{Member,~IEEE} 
\thanks{
	This work was supported by the National Natural Science Foundation
	of China under Grant 62201242, 62331023, and Young Elite Scientists
	Sponsorship Program by CAST 2022QNRC001. Part of this work will be presented at the 2023 IEEE Global Communications Conference, Kuala Lumpur, Malaysia \cite{zhang2023joint}. \emph{(Corresponding author: Changsheng You.)}
	
	Y. Zhang and C. You are with the Department of Electronic and Electrical Engineering, Southern University of Science and Technology (SUSTech), Shenzhen 518055, China (e-mail: yunpu.zhang@my.cityu.edu.hk; youcs@sustech.edu.cn).}		
}

\maketitle

\begin{abstract} 
Extremely large-scale array (XL-array) has emerged as a promising technology to enhance the spectrum efficiency and spatial resolution in future wireless networks by exploiting massive number of antennas for generating pencil-like beamforming. This also leads to a fundamental paradigm shift from conventional far-field communications towards the new near-field communications. In contrast to the existing works that mostly considered simultaneous wireless information
and power transfer (SWIPT) in the far field, we consider in this paper a new and practical scenario, called \emph{mixed near- and far-field} SWIPT, where energy harvesting (EH) and information decoding (ID) receivers are located in the near- and far-field regions of the XL-array base station (BS), respectively.
Specifically, we formulate an optimization problem to maximize the weighted sum-power harvested at all EH receivers by jointly designing the BS beam scheduling and power allocation, under the constraints on the maximum sum-rate and BS transmit power.
First, for the general case with multiple EH and ID receivers, we propose an efficient algorithm to obtain a suboptimal solution by utilizing the binary variable elimination and successive convex approximation methods.
To obtain useful insights, we then study the joint design for special cases. In particular, we show that when there are multiple EH receivers and one ID receiver, in most cases, the optimal design is allocating a portion of power to the ID receiver for satisfying the rate constraint, while the remaining power is allocated to one EH receiver with the highest EH capability. This is in sharp contrast to the conventional far-field SWIPT case, for which all powers should be allocated to ID receivers.
Numerical results show that our proposed joint design significantly outperforms other benchmark schemes without the optimization of beam scheduling and/or power allocation.
\end{abstract}

\section{Introduction}
Future sixth-generation (6G) wireless networks are envisioned to support ever-increasing demands for ultra-high data rate, hyper-reliability, extremely-low latency, etc \cite{8869705,9136592}. These requirements, however, may not be fully achieved by existing fifth-generation (5G) technologies,  thus calling for developing new technologies for 6G. For example, it is expected that 6G communications will migrate to higher frequency bands, such as millimeter-wave (mmWave) and terahertz (THz), to exploit the enormous available bandwidth \cite{7959169,NBAOMP}. Besides, by increasing the number of antennas drastically to another order-of-magnitude, extremely large-scale array/surface (XL-array/surface) has been envisioned as a promising technology to achieve super-high spectrum efficiency and spatial resolution \cite{10041977,9903389,9326394,9140329,9424177,9617121}, which greatly compensates the severe path-loss in high-frequency bands.

With the significant increase of carrier frequency and number of antennas, the well-known \emph{Rayleigh distance} is going to expand to dozens or even hundreds
of meters. This thus leads to a fundamental paradigm shift in the electromagnetic (EM) field characteristics, from the conventional far-field communications towards the new \emph{near-field communications} \cite{eamaz2023near,liu2023near}.
 Specifically, different with far-field communications whose EM field can be approximated as planar waves, the radio propagation in near-field communications should be accurately modeled by spherical waves. This unique channel characteristic enables a new function of \emph{beam focusing}, which focuses the beam energy on a specific location (region),  rather than a specific
direction as in conventional far-field communications \cite{9738442}. 
To exploit the beam-focusing gain of XL-arrays, near-field channel estimation/beam training is indispensable, which, however, is much more challenging than that in far-field communications. For example, the near-field beam training requires a two-dimensional beam
search over both the angular and distance (or range) domains, which significantly differs from the far-field case that  requires one-dimensional angular-domain beam search only.
{\color{black}As such, applying the discrete Fourier transform (DFT)-based far-field
codebook for near-field beam training will lead to degraded training accuracy and rate performance. This can be explained by the so-called \emph{near-field energy-spread} effect, where the energy of a far-field beam steered towards a specific direction  will be dispersed into multiple angles in the near-field, hence making it difficult to determine the true user angle based on the maximum received signal power. To address this issue, a new \emph{polar-domain} codebook was proposed in \cite{9693928}, with each beam
codeword pointing towards a specific location with a targeted
angle and distance. One straightforward beam training method using this codebook is to conduct a two-dimensional exhaustive search over all possible beam codewords. However, this method incurs prohibitively high beam training overhead and hence leaving less time
for data transmission. To reduce the training overhead, an efficient two-phase beam training method was proposed in \cite{9913211} that first estimates the spatial angle and then determines the user distance. This method relies on a key fact that the true
user spatial angle approximately lies in the middle of
a dominant angular region with sufficiently high received signal powers
when using far-field DFT beams for beam training. The required training overhead can be further reduced by designing hierarchical near-field beam training \cite{wu2023two} and exploiting deep learning methods for beam training \cite{9903646}.}

{While the existing works have mainly considered either the near- or far-field communications, the \emph{mixed} near- and far-field communications are also likely to appear, where
there exist both near- and far-field users in the network \cite{you2023near,han2023cross,yu2022adaptive,9598863,10129111,zhang2023mixed}. For instance, consider a base station (BS) that has an 
antenna of diameter $0.5$ meter (m) and operates at a frequency of $30$ GHz. The well-known Rayleigh
distance in this scenario is about $50$ m. This means that in a typical communication
scenario, the users may be located in both the near- and far-field regions from the BS, hence leading to more complicated design issues such as mixed/hybrid-field channel estimation \cite{yu2022adaptive,9598863}, coexistence and interference management of mixed-field users \cite{10129111,zhang2023mixed}, etc.} To be more specific, an interesting observation was revealed in \cite{zhang2023mixed} that due to the energy-spread effect, the near-field user may suffer strong interference from the DFT-based far-field beam,
when its spatial angle is in the neighborhood of the far-field user angle. On the other hand, such power leakage from the DFT-based far-field beam can also be utilized to efficiently charge the near-field user, leading to the new application of mixed-field simultaneous wireless information and power transfer (SWIPT) where energy harvesting (EH) and information decoding (ID) receivers are located in the near- and far-field, respectively. {However, the design of the mixed-field SWIPT also faces several new challenges. For example, by exploiting the near-field beam-focusing property, the beamforming for the near-field EH receivers should be delicately designed to maximize the EH efficiency, while at the same time, minimize the interference to the far-field ID receivers. Second, the energy-spread effect should be accounted for when designing the beamforming for the far-field ID receivers, which can opportunistically charge the near-field EH receivers when they are located in similar angles. Moreover, the power allocation of the BS should be carefully designed to balance the new \emph{near-and-far} beamforming tradeoff in the mixed SWIPT, {with the effects of both beam focusing and energy spread taken into account.}}



\begin{figure}[t]
	\centering
	\includegraphics[width= 0.4\textwidth]{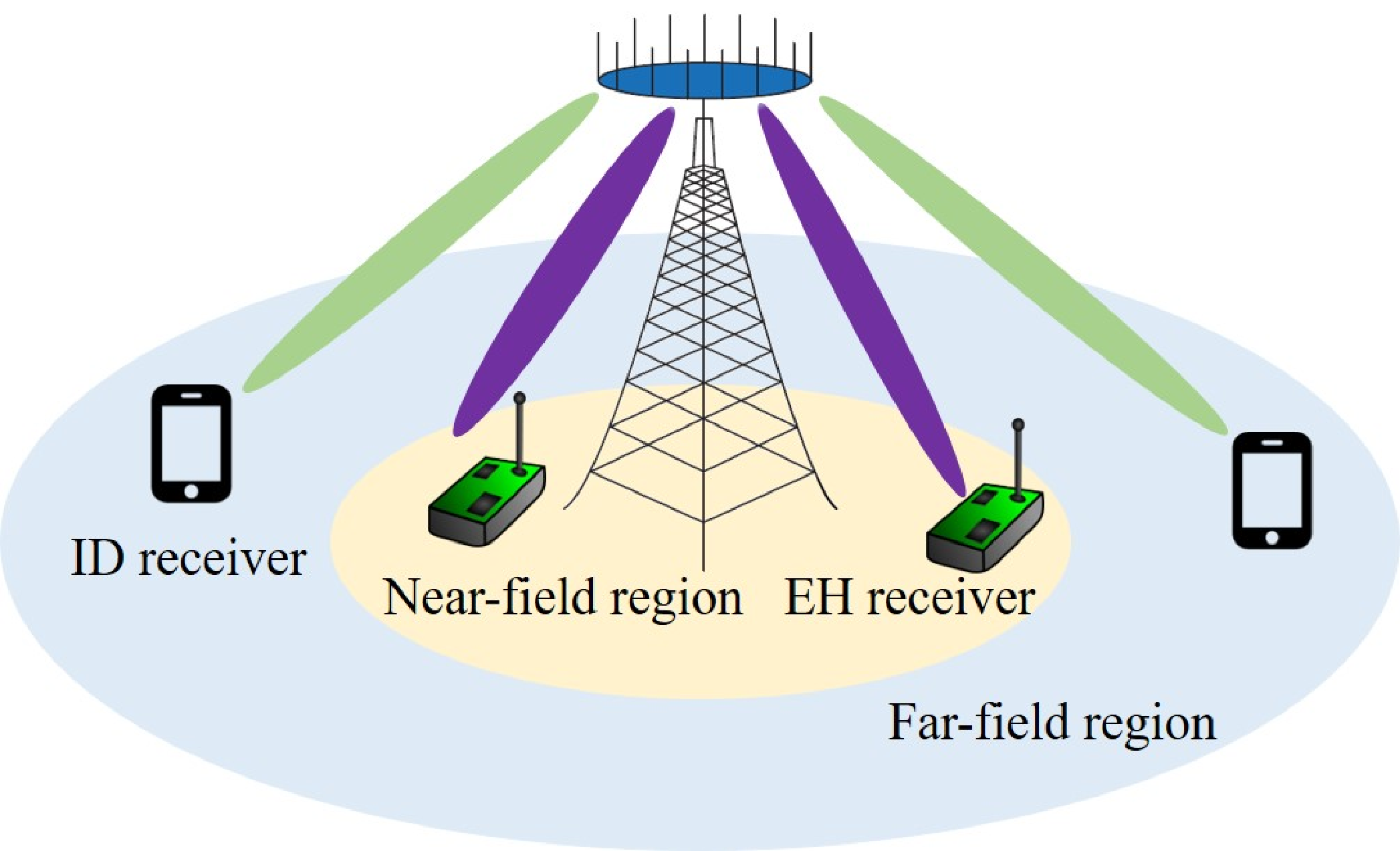}
	\caption{A mixed near- and far-field SWIPT system.}\label{fig:SM}
\end{figure}

To tackle the above issues, we consider in this paper a new \emph{mixed near- and far-field} SWIPT system as shown in Fig. 1, in which the BS equipped with an XL-array simultaneously serves multiple EH and ID receivers, which are located in the near- and far-field regions of the XL-array, respectively. We study the joint BS beam scheduling and power allocation design, where the BS is assumed to simultaneously steer multiple beams towards selected EH and ID receivers with properly designed power allocation. 
The main contributions of this paper are summarized as follows.
\begin{itemize}
	\item {First, to our best knowledge, we are the first one to study the SWIPT in mixed near- and far-field channels, for which we formulate an optimization problem to maximize the weighted sum-power harvested at all EH receivers subject to the constraints on the sum-rate and BS transmit power.} However, this problem is non-convex due to the binary
	beam-scheduling variables as well
	as the intrinsic coupling between beam scheduling and power allocation. To address this issue, we reformulate this problem into a more compact form by using a binary variable elimination method. 
	 Moreover, we show that in the mixed-field SWIPT, the signals for far-field ID receivers can opportunistically charge the near-field EH receivers within a certain angular region, while they receive limited interference from EH receivers as they are far from the BS. 
	\item {Second, for the general case, we propose an efficient algorithm to obtain a suboptimal solution by using the successive convex approximation (SCA) method. In particular, we shed important insights into the joint beam scheduling and power allocation design for some special cases. Specifically, for the case with ID receivers only, the optimal design is shown to exhibit a \emph{water-filling} structure. Next, for the case with EH receivers only, it is shown that allocating all power to the one EH receiver with the highest EH capacity is optimal. Moreover, for the case with multiple EH receivers and one ID receiver, we unveil that the optimal design for the mixed-field SWIPT is allocating a portion of power to the ID receiver for satisfying its rate constraint, while the remaining power is allocated to one EH receiver with the highest EH capability. This is in sharp contrast to the conventional far-field SWIPT case where all powers should be allocated to ID receivers.}
	\item Last, numerical results are provided to demonstrate the effectiveness of the proposed beam scheduling and power allocation scheme under various system setups. It is shown that the proposed scheme significantly outperforms other benchmark schemes with the optimizations of beam scheduling and power allocation, especially when the total transmit power is high and/or the required sum-rate is large. 
%
\end{itemize}

The remainder of this paper is organized as follows. Sections \ref{SMandPF} presents the system
model and problem formulation for the mixed-field SWIPT system. In Section \ref{PRF},  we reformulate the weighted sum-power maximization problem into a more compact form to facilitate the algorithm design. In Section \ref{SCandGC}, we first develop an efficient algorithm to obtain a suboptimal solution for the general case, and then consider three special cases to gain useful insights. Numerical results are presented in Section \ref{SE:NR} to
evaluate the performance of the proposed scheme, followed by the
conclusions given in Section \ref{CL}.

\emph{Notations:}
Upper-case and lower-case boldface letters denote matrices and column vectors,
respectively. Upper-case calligraphic letters (e.g., $\mathcal{K}$) denote discrete and finite sets.
The superscripts $(\cdot)^{H}$ and $(\cdot)^{T}$ stand for the 
transpose and Hermitian transpose, respectively. 
For matrices, $[\cdot]_{i,j}$ denotes the $(i,j)$-th element a matrix, while $[\cdot]_{i}$ refers to the $i$-th element a vector.
$|\cdot|$ denotes the absolute value for a real number and the cardinality for a set. $\mathcal{O}(\cdot)$ denotes the standard big-O notation. $\mathbf{1}_{N}$ and $\mathbf{0}_{N}$ denote an all-one vector of length $N$ and an all-zero vector of length $N$, respectively. 

\section{System Model and Problem Formulation}\label{SMandPF}
\subsection{System Model}
We consider a mixed-field SWIPT system as shown in Fig.~\ref{fig:SM}, where a BS equipped with an $N$-antenna XL-array simultaneously serves multiple EH and ID receivers. Specifically, 
$K$ single-antenna EH receivers, denoted by $\mathcal{K}=\{1,2,\cdots,K\}$, are located in the near-field region of the XL-array for enabling high-power energy harvesting, while $M$ single-antenna ID receivers, denoted by $\mathcal{M}=\{1,2,\cdots,M\}$, are assumed to lie in the far-field region of the XL-array with the BS-ID receiver distances larger than the Rayleigh distance, defined as $Z=\frac{2D^2}{\lambda}$ with $D$ and $\lambda$ denoting the antenna array aperture and carrier wavelength, respectively.\footnote{{For the case where ID receivers locate within the Rayleigh distance, the corresponding sum-rate maximization problem has been investigated in the literature, see, e.g.,  \cite{9738442}.} Besides, consider the following two scenarios. 1) Both EH and ID receivers are in the near field: in this case, the obtained insights in this paper for mixed-field SWIPT may no longer hold for near-field SWIPT due to the vanish of the energy-spread effect and the corresponding near-and-far issue, whose design is left for future work. 2) EH receivers are in the far field and ID receivers are in the near field: in this case, the EH efficiency of far-field EH receivers will drop drastically caused by the long distance, and the near-field ID receivers will suffer severe interference incurred by the far-field beams, which thus is in general impractical.} {For the XL-array BS, it applies the hybrid beamforming architecture to serve $(K+M)$ EH and ID receivers with $N_{\rm RF}$ radio frequency (RF)
	chains, where $N_{\rm RF}\ge K+M$.}

\subsubsection{Near- and Far-Field Channel Models}
In the following, we introduce the channel models
for the near-field EH receivers and far-field ID receivers, respectively.

\underline{\bf Near-field channel model for EH receivers:}
For the near-field EH receiver $k$, its channel from the XL-array BS can be modeled as
\begin{equation}\label{Eq:mp}
	\mathbf{h}^{\rm EH}_{k}=	\mathbf{h}^{\rm EH}_{{\rm LoS}, k}+	\sum_{\ell=1}^{L_{k}}\mathbf{h}^{\rm EH}_{{\rm NLoS}, k, \ell}, ~~~k\in\mathcal{K},
\end{equation}  
where there exist one line-of-sight  (LoS) path and $L_{k}$  non-LoS (NLoS) paths between the BS and EH receiver $k$.
\begin{figure}[t]
	\centering
	\includegraphics[width= 0.45\textwidth]{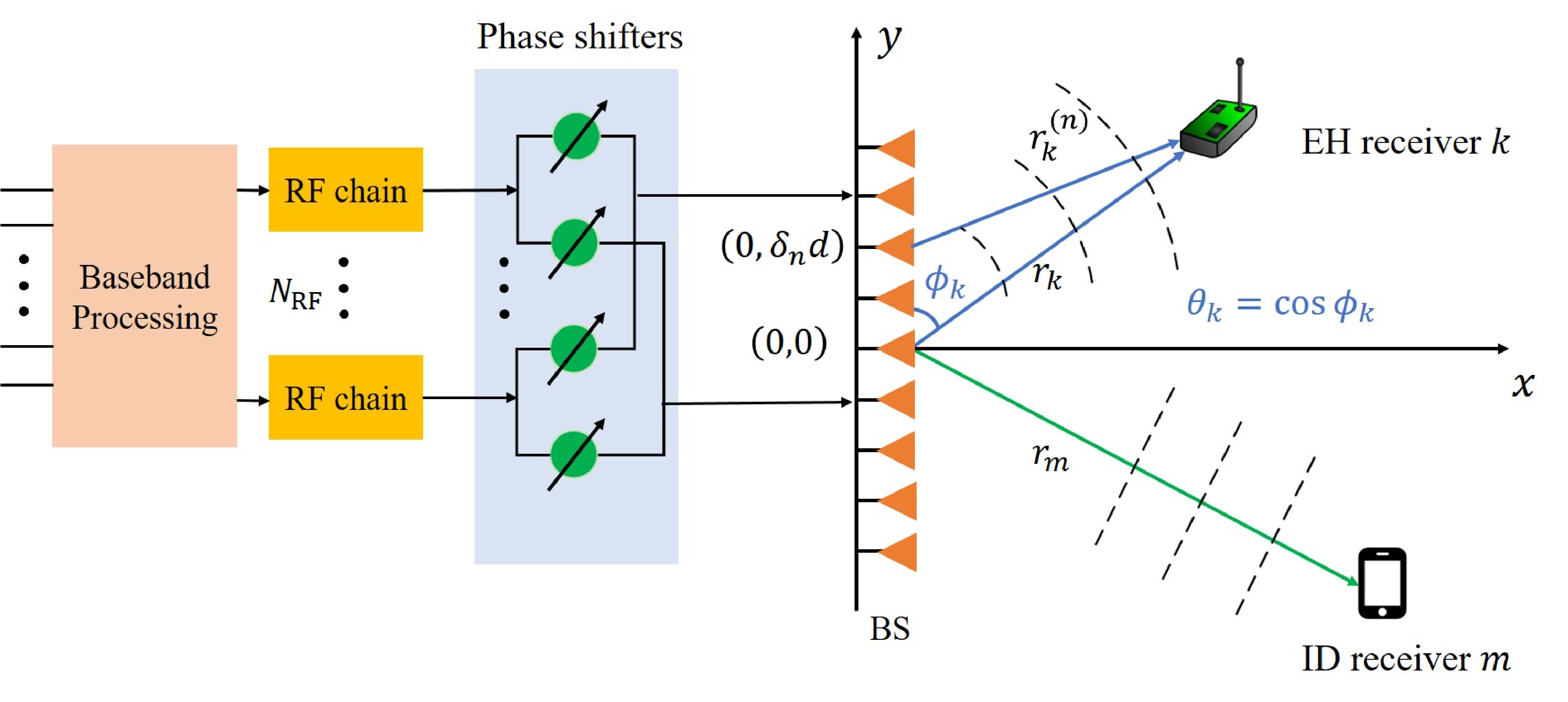}
	\caption{Illustration of the near and far-field channel models.}\label{fig:An}
\end{figure}
In this paper, we consider SWIPT in high-frequency bands (e.g., mmWave and THz), whose channels are susceptible to blockage. Therefore, we mainly consider the LoS channel component for both EH and ID receivers, while the NLoS components are neglected due to small power \cite{9957130,liu2022low,zhang2022near}. As such, the channel from the BS to EH receiver $k$ can be approximated as $\mathbf{h}^{\rm EH}_{k} \approx \mathbf{h}^{\rm EH}_{{\rm LoS}, k}$, which is modeled as follows. {First, without loss of generality, the center of the XL-array BS is assumed to lie in the origin of the coordinate system. Accordingly,  as shown in Fig.~\ref{fig:An}, the coordinate of the $n$-th antenna of the XL-array is located at ($0, \delta_{n}d$) m, where $\delta_{n}=\frac{2n-N+1}{2}$ with $n=0,1,\cdots,N-1$ and $d=\frac{\lambda}{2}$ represents the antenna spacing.} Therefore, based on the near-field spherical wavefront model, the distance between the $n$-th antenna at the BS and EH receiver $k$ is given by
\begin{align}\label{Eq:rrr}
	r^{(n)}_{k}&=\sqrt{r_{k}^2+\delta_{n}^2d^2-2r_{k}\theta_{k} \delta_{n}d},
\end{align}
where $r_{k}$ denotes the distance between the BS antenna center and EH receiver $k$, and $\theta_{k}=2d\cos(\phi_{k})/ \lambda$ denotes the spatial angle at the BS with $\phi_{k}$ denoting the physical angle-of-departure (AoD) from the BS center to EH receiver $k$. Then, by using the second-order Taylor expansion $\sqrt{1+x}\approx1+\frac{1}{2}x-\frac{1}{8}x^2$, $r^{(n)}_{k}$ in \eqref{Eq:rrr} can be approximated as $r^{(n)}_{k}\approx r_{k}-\delta_{n}d\theta_{k}+\frac{\delta_{n}^2d^2(1-\theta^2_{k})}{2r_{k}}$, which is accurate enough when the BS and EH receiver distance $r_{k}$ is smaller than the Rayleigh distance \cite{7942128}.
As such, the LoS channel between the $n$-th BS antenna and EH receiver $k$ can be modeled as $[\mathbf{h}^{\rm EH}_{k}]_{n}=h^{\rm EH}_{k,n}e^{-j 2 \pi(r^{(n)}_{k}-r_{k})/\lambda},$
where $h^{\rm EH}_{k,n}=\frac{\lambda}{4\pi r^{(n)}_{k}}$ is the antenna-wise complex-valued channel gain of EH receiver $k$. Moreover, we assume that the EH receivers are located in the radiative  Fresnel region for which $r_{k}>r_{\rm min}=\max\{\frac{1}{2}\sqrt{\frac{D^3}{\lambda}},1.2D\}$ \cite{bjornson2021primer,7942128}. This condition can be easily satisfied in practice. For example, for an  
XL-array BS which is equipped with an antenna of diameter $D=0.5$ m and operates at a frequency of $30$ GHz, $r_{k}$ is very likely to be larger than $r_{\rm min}=\max\{\frac{1}{2}\sqrt{\frac{D^3}{\lambda}},1.2D\}=1.768$ m. Under the above assumption, we have $h^{\rm EH}_{k,1}\approx h^{\rm EH}_{k,2}\cdots\approx h^{\rm EH}_{k,N}\triangleq h^{\rm EH}_{k}=\frac{\lambda}{4\pi r_{k}}$, where $h^{\rm EH}_{k}$ is the common complex-valued channel gain for different antennas \cite{9738442}.

Based on the above, the LoS-dominant near-field channel from the BS to EH receiver $k$ can be simply modeled as
\begin{equation}
	\mathbf{h}^{\rm EH}_{k}\approx\sqrt{N}h^{\rm EH}_{k} \mathbf{b}(\theta_{k},r_{k}), ~~~k\in\mathcal{K},
\end{equation}  
where $\mathbf{b}(\theta_{k},r_{k})$ denotes the (normalized) near-field channel steering vector, given by
\begin{equation}\label{near_steering}
	\mathbf{b}\left(\theta_{k}, r_{k}\right)\!=\!\!\frac{1}{\sqrt{N}}\!\left[e^{-j 2 \pi(r^{(0)}_{k}-r_{k})/\lambda}, \cdots, e^{-j 2 \pi(r^{(N-1)}_{k}-r_{k})/\lambda}\right]^{T}.
\end{equation}

{\color{black}
\underline{\bf Far-field channel model for ID receivers:} For each ID receiver that is located in the far-field of the BS, say ID receiver $m$, its channel from the BS can be characterized as below based on the planar wavefront propagation model\cite{10159017}
\begin{equation}
	\mathbf{h}^{\rm ID}_{m}=\mathbf{h}^{\rm ID}_{{\rm LoS}, m} +\sum_{\ell=1}^{L_{m}}\mathbf{h}^{\rm ID}_{{\rm NLoS}, m, \ell}, ~~~m\in\mathcal{M},
\end{equation}
where there exist one LoS path and $L_{m}$ NLoS paths between the BS and ID receiver $m$.
By ignoring the negligible NLoS components in high-frequency bands, the BS$\to$ID receiver $m$ channel can be approximated by its LoS component, modeled as
\begin{equation}
	\mathbf{h}^{\rm ID}_{m}\approx\sqrt{N}h^{\rm ID}_{m} \mathbf{a}(\theta_{m}), ~~~m\in\mathcal{M},
\end{equation}
where $h^{\rm ID}_{m}=\frac{\lambda}{4\pi r_{m}}$ represents the complex-valued channel gain of ID receiver $m$. In addition, $\mathbf{a}(\theta_{m})$ denotes the (normalized) far-field channel steering vector, given by
\begin{equation}\label{far_steering}
	\mathbf{a}(\theta_{m})\triangleq \frac{1}{\sqrt{N}}\left[1, e^{j \pi \theta_{m}},\cdots, e^{j \pi (N-1)\theta_{m}}\right]^T,
\end{equation}
where $\theta_{m}=2d\cos(\phi_{m})/ \lambda$ denotes the spatial angle at the BS with $\phi_{m}$ denoting the physical AoD from the BS center to ID receiver $m$.}
\subsubsection{Signal Model}
	Let $x^{\rm EH}_{k}$, $k\in\mathcal{K}$ denote the transmitted energy-carrying signal for each EH receiver $k$ with power $P^{\rm EH}_{k}$ and $x^{\rm ID}_{m}$, $m\in\mathcal{M}$ the transmitted information-bearing signal for each ID receivers $m$  with power $P^{\rm ID}_{m}$. Then, by applying hybrid beamforming, the transmitted signal vector by the BS is given by
\begin{equation}
	\bar{\mathbf{x}}=\mathbf{F}_{\rm A}\mathbf{F}_{\rm D}\mathbf{x},
\end{equation}
where $\mathbf{x}=[x_{1}^{\rm EH},\cdots,x_{K}^{\rm EH},x_{1}^{\rm ID},\cdots,x_{M}^{\rm ID}]^{T}$, $\mathbf{F}_{\rm D}$ represents a $(K+M)\times(K+M)$ digital precoder and {$\mathbf{F}_{\rm A}=[\mathbf{v}^{\rm EH}_{1},\cdots,\mathbf{v}^{\rm EH}_{K},\mathbf{v}^{\rm ID}_{1},\cdots,\mathbf{v}^{\rm ID}_{M}]$ denotes an $N\times(K+M)$ analog precoder with $\mathbf{v}^{\rm EH}_{k}$ and $\mathbf{v}^{\rm ID}_{m}$ representing the analog beamforming vector for EH receiver $k$ and ID receiver $m$, respectively.}
 In this paper, to obtain useful insights as well as reduce the hardware cost of the XL-array, we mainly consider the purely analog beamforming design to evaluate the performance gain achieved by analog beamforming only, for which the digital precoder is set as an identity matrix, i.e., $\mathbf{F}_{\rm D}=\mathbf{1}_{(K+M),(K+M)}$ \cite{8642953}.
 {It is worth noting that, to further improve the performance, the weighted minimum mean square error (WMMSE) or zero-forcing (ZF) based digital beamforming can be properly designed given the analog precoder $\mathbf{F}_{\rm A}$ \cite{9786780,7307218}. The corresponding performance of hybrid beamforming will be evaluated by simulations in Section~\ref{DBBB}.}

Let $\mathcal{D}\triangleq  \{s^{\rm EH}_{1},\cdots,s^{\rm EH}_{K},s^{\rm ID}_{1},\cdots,s^{\rm ID}_{M}\}\in\mathbb{C}^{K+M}$ denote the beam-scheduling indicator set for the XL-array BS, where $s^{\rm EH}_{k}$ and $s^{\rm ID}_{m}$ denote respectively the binary scheduling variable for each EH receiver $k$ and ID receiver $m$. Specifically, $s^{\rm EH}_{k}=1$ if EH receiver $k$ is scheduled by the BS and $s^{\rm EH}_{k}=0$ otherwise; while $s^{\rm ID}_{m}$ is defined in a way similar to $s^{\rm EH}_{k}$.
%

\underline{\bf Signal model for far-field ID receivers:}
Consider the data transmission to	 a far-field ID receiver $m$. Its received signal is given by 
\begin{align}
	y^{\rm ID}_{m}=(\mathbf{h}^{\rm ID}_{m})^H&\mathbf{v}^{\rm ID}_{m}s^{\rm ID}_{m}{x}^{\rm ID}_{m} +\underbrace{(\mathbf{h}^{\rm ID}_{m})^H\sum_{k=1}^{K}\mathbf{v}^{\rm EH}_{k}s^{\rm EH}_{k}{x}^{\rm EH}_{k}}_{\rm \textbf{Interference from EH signals}}\nn\\ 
	&+ \underbrace{(\mathbf{h}^{\rm ID}_{m})^H\sum^{M}_{j=1,j\neq m}\mathbf{v}^{\rm ID}_{j}s^{\rm ID}_{j}{x}^{\rm ID}_{j}}_{\rm \textbf{Interference from other ID signals}}+z^{\rm ID}_{m},
\end{align}
where $z^{\rm ID}_{m}$ is the received additive white Gaussian noise (AWGN) at ID receiver ${m}$ with zero mean and power $\sigma^2_{m}$.
As such, the received signal-to-interference-plus-noise ratio
(SINR) at ID receiver ${m}$ is given by \eqref{Eq:SINR} on the bottom of the next page,
\begin{figure*}[b]
	\hrulefill
	\begin{equation}\label{Eq:SINR}
		{\rm SINR}^{\rm ID}_{m}=\frac{s^{\rm ID}_{m}P^{\rm ID}_{m}g^{\rm ID}_{m}|\mathbf{a}^{H}\left(\theta_{m}\right)\mathbf{v}^{\rm ID}_{m}|^2}{\sum^{K}_{k=1}s^{\rm EH}_{k}P^{\rm EH}_{k}g^{\rm ID}_{m}|\mathbf{a}^{H}\left(\theta_{m}\right)\mathbf{v}^{\rm EH}_{k}|^2+\sum^{M}_{j=1,j	\neq m}s^{\rm ID}_{j}P^{\rm ID}_{j}g^{\rm ID}_{m}|\mathbf{a}^{H}\left(\theta_{m}\right)\mathbf{v}^{\rm ID}_{j}|^2+\sigma^2_{m}},
	\end{equation}
\end{figure*}
where $g^{\rm ID}_{m}=N|h^{\rm ID}_{m}|^2$. The corresponding achievable rate in bits per second per Hertz (bps/Hz) is given by $R^{\rm ID}_{m}=\log_2\left(1+{\rm SINR}^{\rm ID}_{m}\right).$

\underline{\bf Signal model for near-field EH receivers:}
For wireless power transfer (WPT), due to the broadcast property of wireless channels, each EH receiver can harvest wireless energy from both the energy and information signals. As a result, by ignoring the negligible noise power at the
EH receivers and assuming the linear EH model,\footnote{{For simplicity, we consider the widely used linear EH model in the existing literature \cite{6489506,6860253,8941080,9110849}, while the extension to the case under a nonlinear EH model will be discussed in Remark \ref{Nonlinear}.}} the harvested power at EH receiver $k$ is given by
\begin{align}\label{Eq:HP}
	\begin{aligned}
	Q_{k}=\zeta\bigg(s^{\rm EH}_{k}|(&\mathbf{h}^{\rm EH}_{k})^H\mathbf{v}^{\rm EH}_{k}{x}^{\rm EH}_{k}|^2\\
	&+\underbrace{\sum_{i=1,i	\neq k}^{K}s^{\rm EH}_{i}|(\mathbf{h}^{\rm EH}_{k})^H\mathbf{v}^{\rm EH}_{i}{x}^{\rm EH}_{i}|^2}_{\rm \textbf{Harvested power from other EH signals }}\\
	&+\underbrace{\sum_{m=1}^{M}s^{\rm ID}_{m}|(\mathbf{h}^{\rm EH}_{k})^H\mathbf{v}^{\rm ID}_{m}{x}^{\rm ID}_{m}|^2}_{\rm \textbf{Harvested power from ID signals}}\bigg),
\end{aligned}
\end{align}
where $0<\zeta \leq 1$ denotes the energy harvesting efficiency.

\subsection{Problem Formulation}

In this paper, we assume that the BS has perfect channel state information (CSI) of all EH and ID receivers, i.e., near- and far-field channel steering vectors $\mathbf{a}$ and $\mathbf{b}$ as well as channel gains $g^{\rm EH}_{k}$ and $g^{\rm ID}_{m}$. {In practice, this CSI can be efficiently obtained by using existing far-field and near-field channel estimation and beam training methods (see, e.g., \cite{8949454,9129778,9913211,9133142}).} In addition, for ease of implementation, we assume that if an EH or ID receiver is scheduled, the BS will steer a beam towards it to maximize its received power based on e.g., codebook based beamforming, i.e., $\mathbf{v}^{\rm ID}_{m}=\mathbf{a}\left(\theta_{m}\right)$ and $\mathbf{v}^{\rm EH}_{k}=\mathbf{b}\left(\theta_{k}, r_{k}\right)$.\footnote{\color{black}{In this paper, we mainly study the beam scheduling design and thus assume the MRT-based analog beamforming, while the general analog beamforming design deserves further investigation in the future.}} Under the above assumptions, the achievable rate of each ID receiver ${m}$ can be rewritten as \eqref{Eq:sum_rate} shown on the bottom of the next page,
\begin{figure*}[b]
	\hrulefill
	\begin{align}\label{Eq:sum_rate}
		\begin{aligned}
			R^{\rm ID}_{m}&\left(\mathcal{D},\{P^{\rm EH}_{k}\},\{P^{\rm ID}_{m}\}\right)=\log_2\left(1+{\rm SINR}^{\rm ID}_{m}\right)\\
			&=\log_2\left(1+\frac{s^{\rm ID}_{m}P^{\rm ID}_{m}g^{\rm ID}_{m}}{\sum^{K}_{k=1}s^{\rm EH}_{k}P^{\rm EH}_{k}g^{\rm ID}_{m}|\mathbf{a}^{H}\left(\theta_{m}\right)\mathbf{b}\left(\theta_{k}, r_{k}\right)|^2+\sum^{M}_{j=1,j	\neq m}s^{\rm ID}_{j}P^{\rm ID}_{j}g^{\rm ID}_{m}|\mathbf{a}^{H}\left(\theta_{m}\right)\mathbf{a}\left(\theta_{j}\right)|^2+\sigma^2_{m}}\right),
		\end{aligned}
	\end{align}
\end{figure*}
and the harvested power at each EH receiver $k$ is given by 
	\begin{align}\label{Eq:objecfuncx}
	\begin{aligned}
		Q_{k}&\left(\mathcal{D},\{P^{\rm EH}_{k}\},\{P^{\rm ID}_{m}\}\right)=
		\zeta\bigg(s^{\rm EH}_{k}P^{\rm EH}_{k}g^{\rm EH}_{k}\\
		&+\sum_{i=1,i	\neq k}^{K} s^{\rm EH}_{i}P^{\rm EH}_{i}g^{\rm EH}_{k}|(\mathbf{b}^H(\theta_{k},r_{k})\mathbf{b}(\theta_{i},r_{i})|^2\\
		&+\sum_{m=1}^{M}s^{\rm ID}_{m}P^{\rm ID}_{m}g^{\rm EH}_{k}|\mathbf{b}^H(\theta_{k},r_{k})\mathbf{a}\left(\theta_{m}\right)|^2\bigg),
	\end{aligned}
\end{align}
where $g^{\rm EH}_{k}=N|h^{\rm EH}_{k}|^2$.

Our objective is to jointly optimize the beam scheduling (i.e., $\mathcal{D}$) and power allocation (i.e., $\{P^{\rm EH}_{k}\}$ and $\{P^{\rm ID}_{m}\}$) of the XL-array BS for maximizing the weighted sum-power harvested at all EH receivers, subject to a sum-rate constraint for all ID receivers and a total BS transmit power constraint.
Let $\alpha_{k} \ge 0$ denote a predefined power weight for each EH receiver $k$, where a larger value of $\alpha_{k}$ indicates higher preference for transferring energy to EH receiver $k$, as compared to other EH receivers. As such, the weighted sum-power transferred to all EH receivers, denoted by $Q$, can be expressed as
\begin{equation}\label{Eq:sum-power}
	\!Q\!\left(\mathcal{D},\{P^{\rm EH}_{k}\},\{P^{\rm ID}_{m}\}\right) \!=\!\sum_{k=1}^{K}\!\alpha_{k}	Q_{k}\!\left(\mathcal{D},\{P^{\rm EH}_{k}\},\{P^{\rm ID}_{m}\}\right).
\end{equation}
Based on the above, this optimization problem can be formulated as follows
 \begin{subequations}
 	\begin{align}
 		({\bf P1}):\max_{\substack{\mathcal{D},\{P^{\rm EH}_{k}\},\\\{P^{\rm ID}_{m}\}} }  &~~	Q\left(\mathcal{D},\{P^{\rm EH}_{k}\},\{P^{\rm ID}_{m}\}\right)
 		\label{Eq:Or_OJ}\\
 		\text{s.t.}~~~
 		&~	\sum_{m=1}^{M}R^{\rm ID}_{m}\left(\mathcal{D},\{P^{\rm EH}_{k}\},\{P^{\rm ID}_{m}\}\right)\ge R,\label{C:sum-rate}
 		\\
 		&~
 		s^{\rm EH}_{k},s^{\rm ID}_{m} \in \{0,1\},~~ k\in \mathcal{K},m\in \mathcal{M},\label{C:binary}\\
 		&~
 		 \sum_{k=1}^{K}s^{\rm EH}_{k}P^{\rm EH}_{k}+\sum_{m=1}^{M}s^{\rm ID}_{m}P^{\rm ID}_{m}\le P_{0},\label{C:Overall_P}\\
 		 &~
 		 P^{\rm EH}_{k} \ge 0, P^{\rm ID}_{m}\ge 0,~k\in \mathcal{K},m\in \mathcal{M},\label{C:nonnegative}
 	\end{align}
 \end{subequations}
{\color{black}where \eqref{C:sum-rate} indicates the sum-rate constraint for all ID receivers; \eqref{C:binary} denotes the binary beam-scheduling indicator for each EH receiver $k$ and ID receiver $m$; {\eqref{C:Overall_P} is the transmit power constraint for the BS with $P_{0}$ denoting the maximum transmit power.}}

Problem (P1) is a mixed-integer optimization problem due to the binary beam-scheduling optimization variables $\mathcal{D}$ in \eqref{C:binary} and the continuous power-allocation variables $\{P^{\rm EH}_{k}\}$, $\{P^{\rm ID}_{m}\}$ in \eqref{C:Overall_P}. Moreover, the beam scheduling and power allocation optimization are strongly coupled in the objective function and the constraints in \eqref{C:sum-rate}--\eqref{C:Overall_P}, which renders problem (P1) more difficult to be optimally solved. Among others, one straightforward method to solve (P1) is by first enumerating all possible EH and ID scheduling combinations with optimized power allocation and then selecting the optimal desgin. However, enumerating all beam scheduling designs
will incur a computational complexity of $2^{K+M}$, which exponentially increases with $K+M$ and thus is unaffordable in practice, especially for a large-scale mixed-field SWIPT system with a large number of EH and ID receivers. To tackle these difficulties, we propose an efficient algorithm in this paper to obtain a high-quality solution to problem (P1) and shed useful insights into the joint beam scheduling and power allocation design.

\section{Problem Reformulation}\label{PRF}
In this section, we reformulate problem (P1) into a more compact form to facilitate the subsequent algorithm design. 

\subsection{Eliminating the Binary Optimization Variables, $s^{\rm EH}_{k}$ and $s^{\rm ID}_{m}$}
One of the main challenges in solving problem (P1) arises from the intrinsic coupling between the binary scheduling variables (i.e., $s^{\rm EH}_{k}$ and $s^{\rm ID}_{m}$) and transmit power allocation (i.e., $\{P^{\rm EH}_{k}\}$ and $\{P^{\rm ID}_{m}\}$), which appear in both the objective function and constraints. To address this issue, we introduce the following continuous variables to eliminate the binary optimization variables
\begin{equation}\label{b_s_e}
	\tilde{P}^{\rm EH}_{k} = s^{\rm EH}_{k}P^{\rm EH}_{k}, ~~~~ \tilde{P}^{\rm ID}_{m} = s^{\rm ID}_{m}P^{\rm ID}_{m}.
\end{equation}
As such, the constraint \eqref{C:sum-rate} can be equivalently transformed into the following form
\begin{align}\label{Eq:C:sum-rate_t}
\!\!\sum_{m=1}^{M}\!R^{\rm ID}_{m}(\mathcal{D},\!\{P^{\rm EH}_{k}\},\!\{P^{\rm ID}_{m}\})\!=\!\sum_{m=1}^{M}\!\tilde{R}^{\rm ID}_{m}(\{\tilde{P}^{\rm EH}_{k}\},\{\tilde{P}^{\rm ID}_{m}\}),
\end{align}
where $\tilde{R}^{\rm ID}_{m}(\{\tilde{P}^{\rm EH}_{k}\},\{\tilde{P}^{\rm ID}_{m}\})$ is given in \eqref{NewR} on the bottom of the next page.
\begin{figure*}[b]
	\hrulefill
\begin{align}\label{NewR}
	\tilde{R}^{\rm ID}_{m}(\{\tilde{P}^{\rm EH}_{k}\},\{\tilde{P}^{\rm ID}_{m}\})\!=\!
	\log_2\left(\!1\!+\!\frac{\tilde{P}^{\rm ID}_{m}g^{\rm ID}_{m}}{\sum^{K}_{k=1}\tilde{P}^{\rm EH}_{k}g^{\rm ID}_{m}|\mathbf{a}^{H}\left(\theta_{m}\right)\mathbf{b}\left(\theta_{k}, r_{k}\right)|^2+\sum^{M}_{j=1,j	\neq m}\tilde{P}^{\rm ID}_{j}g^{\rm ID}_{m}|\mathbf{a}^{H}\left(\theta_{m}\right)\mathbf{a}\left(\theta_{j}\right)|^2+\sigma^2_{m}}\right).
\end{align}
\end{figure*}
Note that there is a one-to-one correspondence between $\tilde{P}^{\rm EH}_{k}$ and $\{s^{\rm EH}_k, {P}^{\rm EH}_{k}\}$. For example, if $\tilde{P}^{\rm EH}_{k}>0$, we have $s^{\rm EH}_k=1$ and ${P}^{\rm EH}_{k}= \tilde{P}^{\rm EH}_{k}$. Otherwise, if $\tilde{P}^{\rm EH}_{k}=0$, we have $s^{\rm EH}_k=0$ and ${P}^{\rm EH}_{k}=0$. 
As such, the objective function in \eqref{Eq:Or_OJ} can be expressed in a simpler form in \eqref{NewQ} on the bottom of the next page.
\begin{figure*}[b]
	\hrulefill
\begin{align}\label{NewQ}
\tilde{Q}(\{\tilde{P}^{\rm EH}_{k}\},\{\tilde{P}^{\rm ID}_{m}\}) = 
	\sum_{k=1}^{K}\alpha_{k}\zeta	\left(\tilde{P}^{\rm EH}_{k}g^{\rm EH}_{k}+\sum_{i=1,i	\neq k}^{K}\!\!\tilde{P}^{\rm EH}_{i}g^{\rm EH}_{k}|(\mathbf{b}^H(\theta_{k},r_{k})\mathbf{b}(\theta_{i},r_{i})|^2+\sum_{m=1}^{M}\tilde{P}^{\rm ID}_{m}g^{\rm EH}_{k}|\mathbf{b}^H(\theta_{k},r_{k})\mathbf{a}\left(\theta_{m}\right)|^2\right).
\end{align}
\end{figure*}
Similarly, the constraint in \eqref{C:Overall_P} can be rewritten as $\sum_{k=1}^{K}\tilde{P}^{\rm EH}_{k}+\sum_{m=1}^{M}\tilde{P}^{\rm ID}_{m}\le P_{0}.$

Based on the above, problem (P1) is equivalently expressed as follows
\begin{subequations}
	\begin{align}
		({\bf P2}):\max_{\substack{\{\tilde{P}^{\rm EH}_{k}\},\\\{\tilde{P}^{\rm ID}_{m}\}} }  &~\tilde{Q}(\{\tilde{P}^{\rm EH}_{k}\},\{\tilde{P}^{\rm ID}_{m}\}) \nn\\
		\text{s.t.}~
		&~ \sum_{m=1}^{M}\tilde{R}^{\rm ID}_{m}(\{\tilde{P}^{\rm EH}_{k}\},\{\tilde{P}^{\rm ID}_{m}\})\ge R,
		\\
		&~
		\sum_{k=1}^{K}\tilde{P}^{\rm EH}_{k}+\sum_{m=1}^{M}\tilde{P}^{\rm ID}_{m}\le P_0,\\
			&~ \tilde{P}^{\rm EH}_{k} \ge 0,\tilde{P}^{\rm ID}_{m}\ge 0, ~k\in \mathcal{K},m\in \mathcal{M}.\label{C:nonnegative1}
	\end{align}
\end{subequations}

\subsection{Correlation Evaluation between EH and ID Receivers}
Before solving problem (P2), we first study the correlation between the channels of EH and ID receivers, so as to shed useful insights into the beam scheduling and power allocation design for the mixed-field SWIPT. To this end, we first make a key definition below.
\begin{definition}\label{De:Correlation} 
	\emph{The correlation between any two near-field steering vectors is
		\begin{align}
			\eta(\theta_{p},\theta_{q},r_{p},r_{q})=|\mathbf{b}^{H}(\theta_{p},r_{p})\mathbf{b}(\theta_{q},r_{q})|.
	\end{align}}	
\end{definition}

\begin{remark}\label{Fresnel}
	\emph{Our prior work \cite{zhang2023mixed} studies the correlation between the near- and far-field steering vectors and reveals that the DFT-based far-field beams may cause strong interference to the near-field user even when they locate in different spatial regions. However, when taking a different view from the EH perspective, the power leakage from the DFT-based far-field beams to the near-field user can be used for charging EH devices efficiently as modeled in \eqref{Eq:HP}. Moreover, it is worth noting that the correlation between near-field steering vectors is a general version of the near-and-far correlation since the far-field channel model is an approximation of the near-field channel model, which is shown below.}
\end{remark}

\begin{figure*}[t]
	\centering
	\subfigure[System setting.]{\includegraphics[width=5cm]{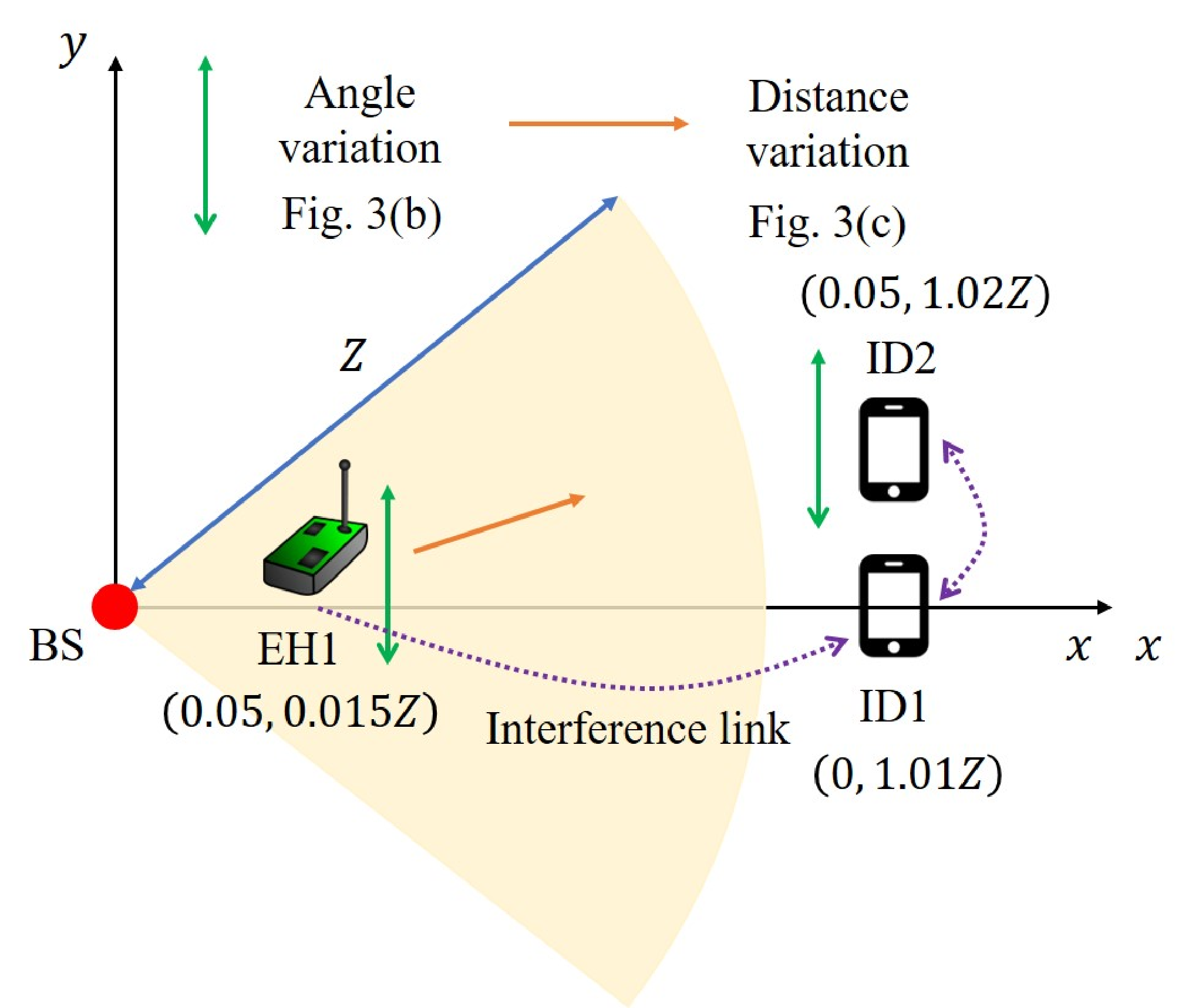}\label{fig:inter}}
	\subfigure[Interference vs. spatial angle.]{\includegraphics[width=4.9cm]{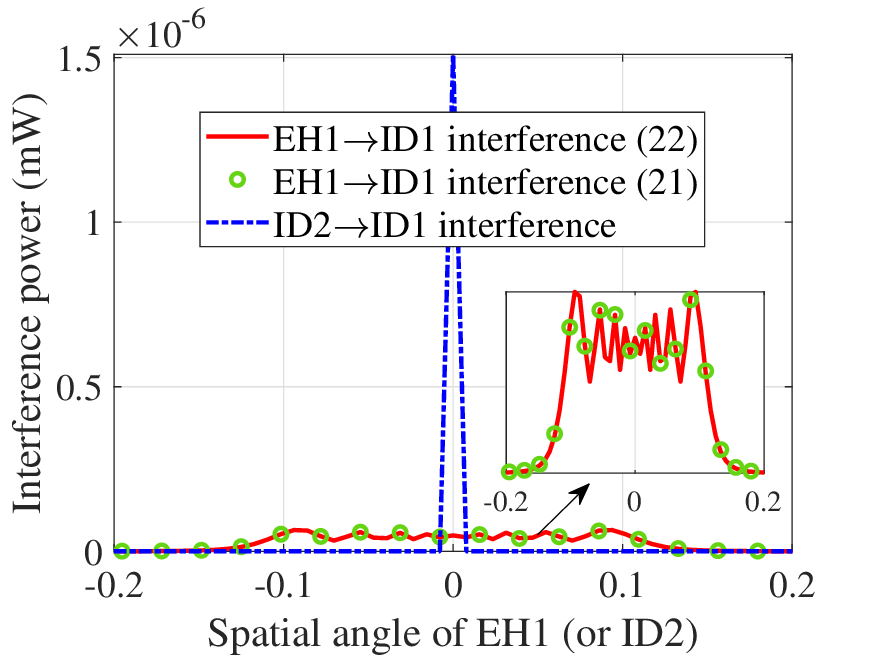}\label{fig:ID_angle}}
	\subfigure[Interference vs. distance.]{\includegraphics[width=4.9cm]{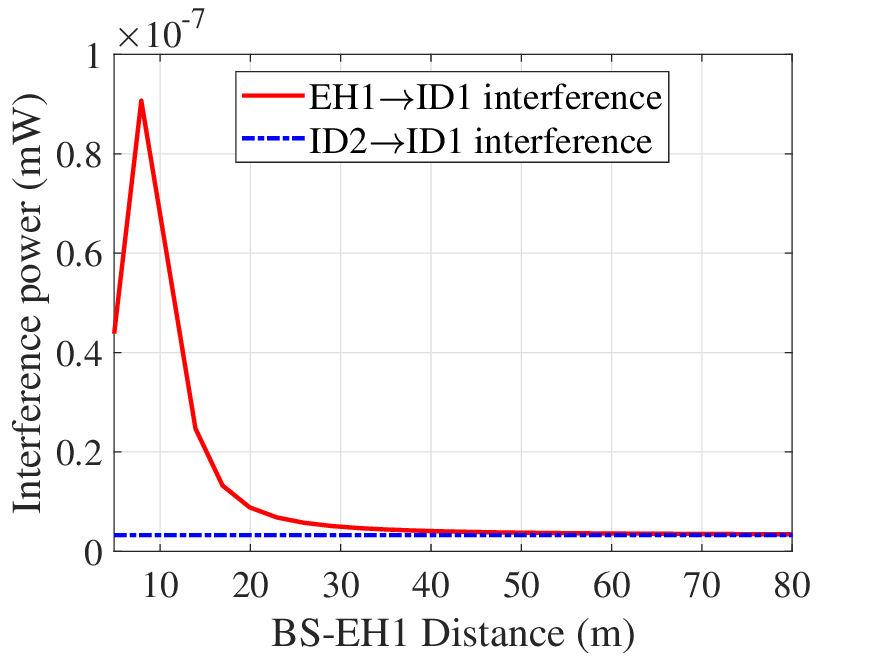}\label{fig:ID_dist}}
	\subfigure[System setting.]{\includegraphics[width=4.9cm]{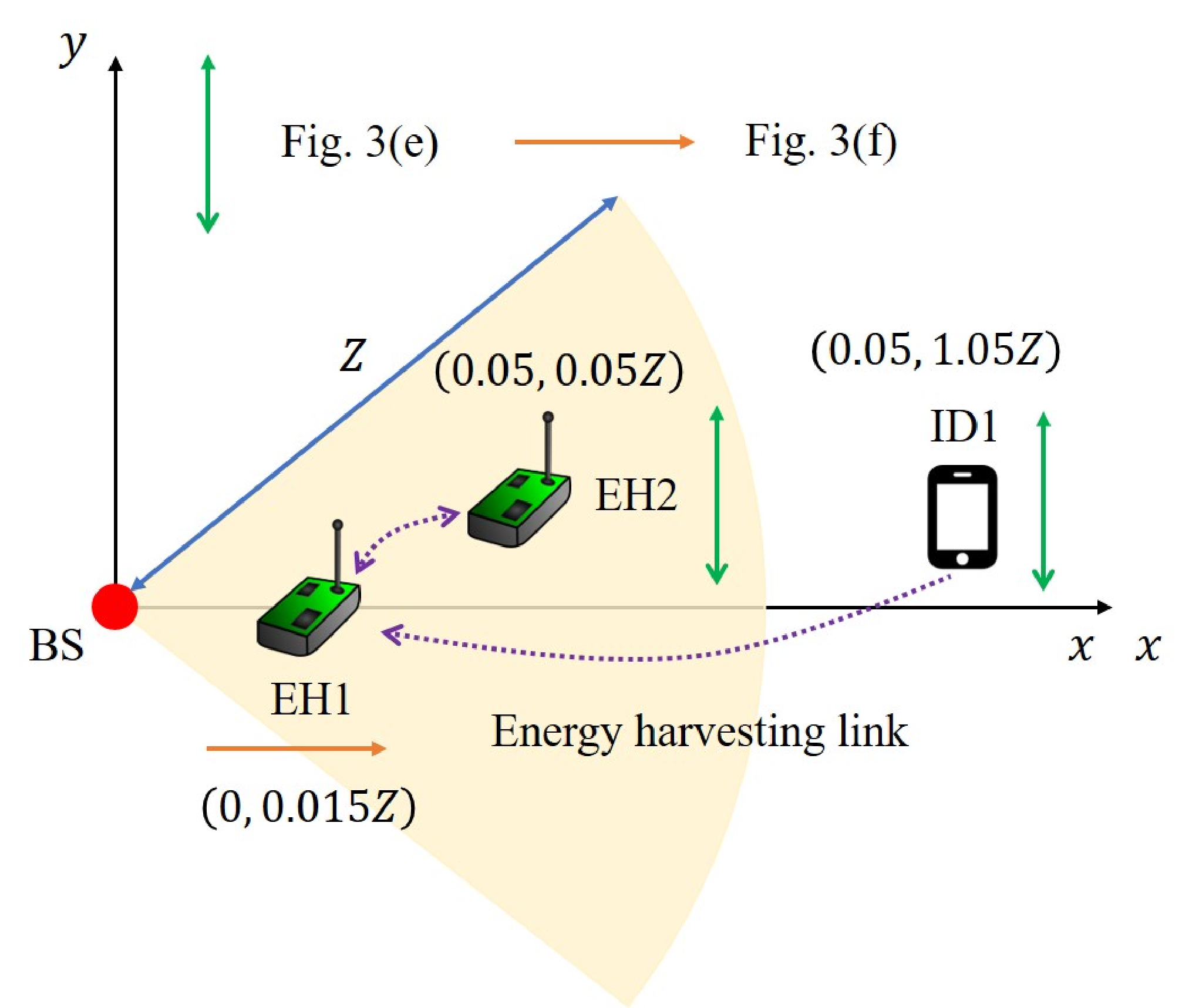}\label{fig:power}}
	\subfigure[Harvested power vs. spatial angle.]{\includegraphics[width=4.9cm]{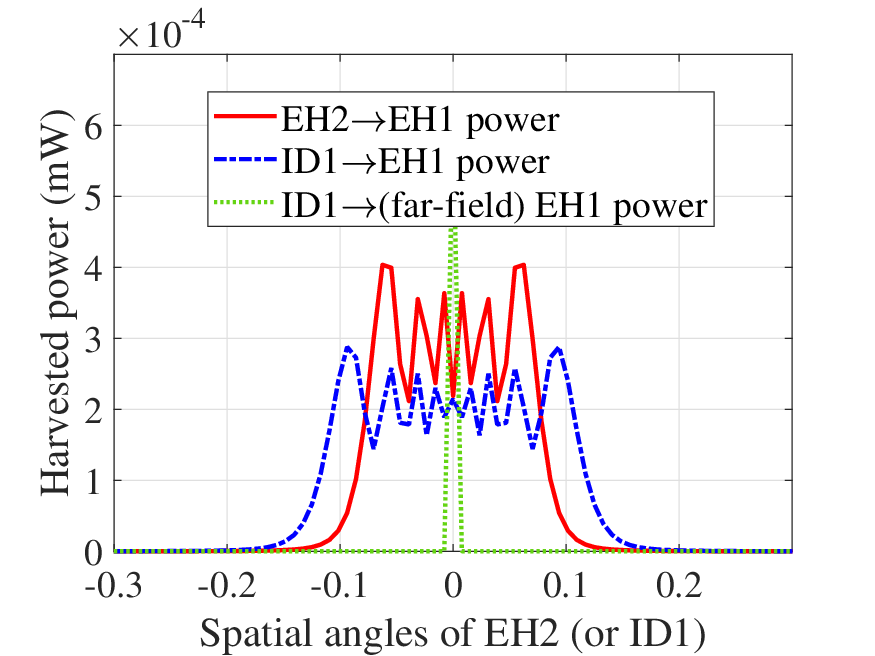}\label{fig:EH_angle}}	
	\subfigure[Harvested power vs. distance.]{\includegraphics[width=4.9cm]{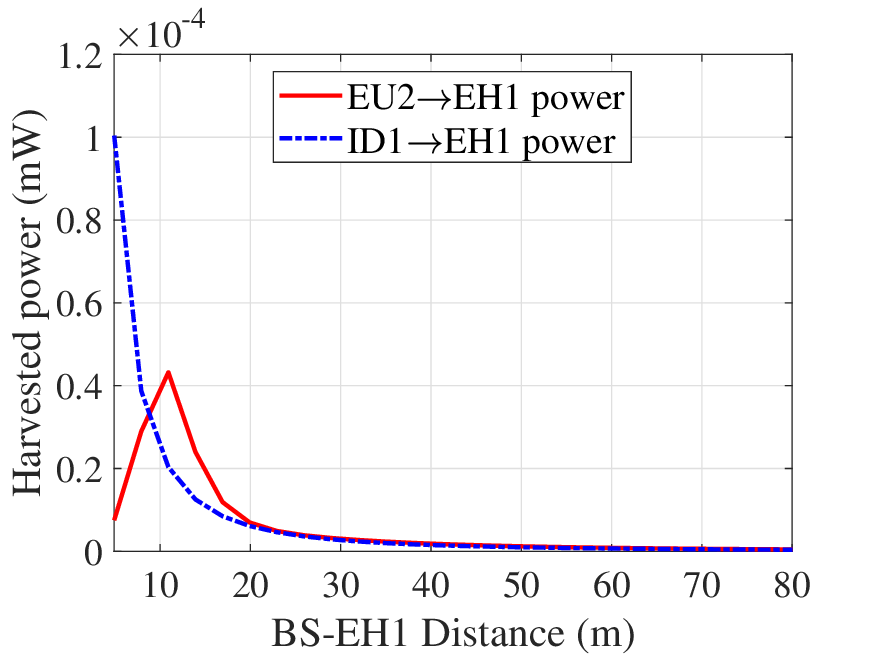}\label{fig:EH_dist}}
	\caption{Illustration of interference power at the far-field ID receivers and harvested power at the near-field receivers.}
\end{figure*}

\begin{lemma}\label{Th:correlation}
	\emph{The correlation between any two near-field steering vectors $	\eta(\theta_{p},\theta_{q},r_{p},r_{q})$ can be approximated as
		\begin{align}\label{NN_correlation}
			\eta(\theta_{p},\theta_{q},r_{p},r_{q})=|\mathbf{b}^{H}&(\theta_{p},r_{p})\mathbf{b}(\theta_{q},r_{q})|\nn\\
			&\approx\left|\frac{\hat{C}(\beta_1,\beta_2)+j\hat{S}(\beta_1,\beta_2)}{2\beta_2}\right|,
		\end{align}
		where 
		\begin{equation}\label{keypara}
			\beta_1=\frac{(\theta_q-\theta_p)}{\sqrt{d\left|\frac{1-\theta_p^2}{ r_p}-\frac{1-\theta_q^2}{  r_q}\right|}},~~~
			\beta_2=\frac{N}{2}\sqrt{d\left|\frac{1-\theta_p^2}{ r_p}-\frac{1-\theta_q^2}{  r_q}\right|}.
		\end{equation}
		Besides, $\hat{C}(\beta_1,\beta_2)={C}(\beta_1+\beta_2)-{C}(\beta_1-\beta_2)$ and $\hat{S}(\beta_1,\beta_2)={S}(\beta_1+\beta_2)-{S}(\beta_1-\beta_2)$, in which ${C}(\beta)=\int_{0}^{\beta}\cos(\frac{\pi}{2}t^2) {\rm d} t$ and ${S}(\beta)=\int_{0}^{\beta}\sin(\frac{\pi}{2}t^2) {\rm d} t$ are the Fresnel integrals.  
	}	
\end{lemma}
\begin{proof}
The proof is similar to that in [Lemma 1,  \cite{chen2022hierarchical}] and hence is omitted for brevity.
\end{proof}

{\color{black}Lemma~\ref{Th:correlation} reveals the correlation of any two near-field steering vectors, which is fundamentally determined by two key parameters $\beta_1$ and $\beta_2$ in \eqref{keypara}. To be more specific, $\beta_1$ is a function of spatial angles (i.e., $\theta_{p}$ and $\theta_{q}$) and distances (i.e., $r_{p}$ and $r_{q}$), while $\beta_2$ depends on the number of BS antennas, as well as the spatial angles and distances. Note that the correlation has a symmetry property, i.e., $\eta(\theta_{p},\theta_{q},r_{p},r_{q})=\eta(\theta_{q},\theta_{p},r_{q},r_{p})$. Moreover, the channel correlation in \eqref{NN_correlation} is a general version of the near-and-far correlation obtained in \cite{zhang2023mixed} by setting $r_{p}=\infty$ or $r_{q}=\infty$, i.e., one near-field channel steering vector reduces to a far-field channel steering vector.}

Next, we present two examples to show interesting insights about the interference and harvested power in the considered mixed-field SWIPT system.
{\color{black}
	\begin{example}[Interference power at the far-field ID receivers]
		\emph{We present a concrete example in Fig.~\ref{fig:inter} to show the interference power at a far-field ID receiver, ID1, from signals for another ID receiver ID2 (located in the far-field) and an EH receiver EH1 (located in the near-field). Specifically, we consider a mixed-field SWIPT system with $N=256$, $f=30$ GHz and $P^{\rm EH}_{1}=P^{\rm ID}_{2}=30$ dBm, where EH1, ID1 and ID2 are respectively located at $(0.05,0.015Z)$, $(0,1.01Z)$ and $(0.05,1.02Z)$ under the polar coordinate ($\theta, r$). We plot in Figs. \ref{fig:ID_angle} and \ref{fig:ID_dist} the interference powers at ID1 versus the EH1 (or ID2) spatial angle and BS-EH1 distance, respectively. First, one can observe from Fig. \ref{fig:ID_angle} that the interference from ID2 is dominant only when ID1 and ID2 are located in the same spatial angle, while there is interference from EH1 within a relatively wide angular region (i.e., ($-0.16,0.16$)). Next, it is shown in Fig. \ref{fig:ID_dist} that when EH1 and ID1 are located in different spatial angles (i.e., $0.05$ vs. $0$), the interference from EH1 to ID1 generally decreases with the BS-EH1 distance, since EH1 gradually reduces to a far-field user when its distance with the BS is larger than $Z$ and there is small interference power when two far-field receivers locate in different spatial angles.}
	\end{example}
	\begin{example}[Harvested power at the near-field EH receivers]\label{Ex:PL}
	\emph{In Fig. \ref{fig:power}, we present an example to show the difference in the harvested power at an EH receiver, EH1, from signals for another EH receiver EH2 (located in the near-field) and an ID receiver ID1 (located in the far-field). In this case, the harvested powers at EH1 versus EH2 (or ID1) spatial angle and BS-EH1 distance are plotted in Figs. \ref{fig:EH_angle} and \ref{fig:EH_dist}, respectively. First, it is observed that the harvested power at EH1 comes from the signals for EH2 and ID1, both of which provide significant power within a certain angular region. This is quite different from the conventional far-field SWIPT case, in which ID1 signal can charge another (far-field) EH receiver only when they are located in similar spatial angles (see the green dotted line in Figs. \ref{fig:EH_angle}). Moreover, as shown in Fig. \ref{fig:EH_dist}, the harvested power at EH1 generally decreases with the BS-EH1 distance, and eventually becomes negligible when the BS-EH1 distance is sufficiently long (i.e., EH1 locates in the far-field region). This is because the harvested power from ID1 becomes smaller when both EH1 and ID1 are located in the far-field region with different spatial angles.}
\end{example}
}
\subsection{Problem Reformulation} 
{\color{black}Based on Lemma \ref{Th:correlation}, $\tilde{R}(\{\tilde{P}^{\rm EH}_{k}\},\{\tilde{P}^{\rm ID}_{m}\})$ in \eqref{NewR} and $\tilde{Q}(\{\tilde{P}^{\rm EH}_{k}\},\{\tilde{P}^{\rm ID}_{m}\})$ in \eqref{NewQ} can be rewritten as functions of the EH/ID correlation, given in \eqref{Eq:NewnewR} and \eqref{Eq:NewnewQ} on the bottom of the page.}
\begin{figure*}[b]
\hrulefill
\begin{align}\label{Eq:NewnewR}
\tilde{R}^{\rm ID}_{m}(\{\tilde{P}^{\rm EH}_{k}\},\{\tilde{P}^{\rm ID}_{m}\})=
\log_2\left(1+\frac{\tilde{P}^{\rm ID}_{m}g^{\rm ID}_{m}}{\sum^{K}_{k=1}\tilde{P}^{\rm EH}_{k}g^{\rm ID}_{m}\eta^2(\theta_{m},\theta_{k},r_{k})+\sum^{M}_{j=1,j	\neq m}\tilde{P}^{\rm ID}_{j}g^{\rm ID}_{m}\eta^2(\theta_{m},\theta_{j})+\sigma^2_{m}}\right),
\end{align}
\end{figure*}
\begin{figure*}[b]
	\hrulefill
\begin{align}\label{Eq:NewnewQ}
	\tilde{Q}(\{\tilde{P}^{\rm EH}_{k}\},\{\tilde{P}^{\rm ID}_{m}\})=
	\sum_{k=1}^{K}\!\alpha_{k}\zeta
	\left(\tilde{P}^{\rm EH}_{k}g^{\rm EH}_{k}+\sum_{i=1,i	\neq k}^{K}\tilde{P}^{\rm EH}_{i}g^{\rm EH}_{k}\eta^2(\theta_{k},\theta_{i},r_{k},r_{i})+\sum_{m=1}^{M}\tilde{P}^{\rm ID}_{m}g^{\rm EH}_{k}\eta^2(\theta_{m},\theta_{k},r_{k})\right).
\end{align}
\end{figure*}

To facilitate the analysis, we first define a correlation matrix as 
\begin{equation}
	\mathbf{\Lambda}=
	\begin{bmatrix}
		\eta^2_{1,1}&\eta^2_{1,2}&\cdots&\eta^2_{1,K+M}\\
		\eta^2_{2,1}&\eta^2_{2,2}&\cdots&\eta^2_{2,K+M}\\
		\vdots&&\ddots&\vdots\\
		\eta^2_{K+M,1}&\eta^2_{K+M,2}&\cdots&\eta^2_{K+M,K+M}\\
	\end{bmatrix},
\end{equation}
where $\eta^2_{p,q}$ denotes the correlation between the channel steering vectors of EH/ID receiver $p$ and EH/ID  receiver $q$, which is given by
\begin{equation}
	\eta^2_{p,q}=
	\begin{cases}
		|\mathbf{b}^{H}(\theta_{p},r_{p})\mathbf{b}(\theta_{q},r_{q})|^2 & \text{if}~1\le p,q \le K, \\
		|\mathbf{b}^{H}(\theta_{p},r_{p})\mathbf{a}(\theta_{q})|^2 & \text{if}~ 1\le p \le K,\\&\quad K+1\le q \le K+M,\\
			|\mathbf{a}^{H}(\theta_{p})\mathbf{a}(\theta_{q})|^2 & \text{if}~ K+1\le p,q\le K+M,\\
				1 &\text{if}~p=q.\\
	\end{cases}
\end{equation}

{\color{black} Next, to rewrite problem (P2) in a more compact form, we further slightly change the form of the correlation matrix $\mathbf{\Lambda}$ by setting some of the
 diagonal elements to be zero (i.e., $([\mathbf{\Lambda}]_{(K+1),(K+1)}=0,\cdots,[\mathbf{\Lambda}]_{(K+M),(K+M)}=0)$), since we do not consider energy harvesting at the ID receivers. As such, the new correlation matrix is given by}
\begin{equation}
	\mathbf{\bar{\Lambda}}=
	\begin{bmatrix}
		1&\eta^2_{1,2}&\cdots&\eta^2_{1,K+M}\\
		\eta^2_{2,1}&1&\cdots&\eta^2_{2,K+M}\\
		\vdots&&\ddots&\vdots\\
		\eta^2_{K+M,1}&\eta^2_{K+M,2}&\cdots&0\\
	\end{bmatrix}.
\end{equation}
Let $\mathbf{y}=\left[\tilde{P}^{\rm EH}_{1},\cdots,\tilde{P}^{\rm EH}_{k},\cdots,\tilde{P}^{\rm EH}_{K},\tilde{P}^{\rm ID}_{1},\cdots,\tilde{P}^{\rm ID}_{m},\cdots,\tilde{P}^{\rm ID}_{M}\right]^{T}$ denote the vector including all power allocation optimization variables. Then, problem (P2) can be equivalently recast as follows
\begin{subequations}
	\begin{align}
		({\bf P3}):~~\max_{\substack{\mathbf{y}} }  &~~~~	(\mathbf{c}^{\rm EH})^{T}\mathbf{\bar{\Lambda}}\mathbf{y}\nn\\
		\text{s.t.}
		&~~~~ \sum_{m=1}^{M}\log_{2}\left( 1+\frac{(	\mathbf{c}^{\rm ID}_{m})^{T}\mathbf{y}}{(	\mathbf{c}^{\rm ID}_{m})^{T}\mathbf{\bar{\Lambda}}\mathbf{y}+\sigma^2_{m}}\right) \ge R,\label{C:sum-rate_m}
		\\
		&~~~~
		\mathbf{1}^{T}_{(K+M)\times 1}\mathbf{y}\le P_{0},\label{C:sum-power1}\\
		&~~~~
		\mathbf{y} \succeq  \mathbf{0}\label{C:nonn2},
	\end{align}
\end{subequations}
where $\mathbf{c}^{\rm EH}\!\!=\!\!\left[g^{\rm EH}_{1},\cdots,g^{\rm EH}_{k},\cdots,g^{\rm EH}_{K}, \mathbf{0}_{1\times M}\right]^{T}$ and $
	\mathbf{c}^{\rm ID}_{m}\!\!=\!\!\left[\mathbf{0}_{1\times K}, \mathbf{0}_{1\times (m-1)},g^{\rm ID}_{m}, \mathbf{0}_{(m+1)\times M}\right]^{T} ,\forall m\in \mathcal{M}.$
	
{\color{black}Problem (P3) is feasible if there is at least one feasible solution satisfying all constraints, especially the sum-rate constraint in \eqref{C:sum-rate_m}. This issue can be solved by formulating a feasibility-check problem, as will shown in Section~\ref{feasibility-check}. Next, compared with problem (P2), 
	problem (P3) renders a more concise form for the considered joint beam scheduling and power allocation optimization, since the objective function and constraints in \eqref{C:sum-power1} and \eqref{C:nonn2} are all affine. However, the sum-rate constraint in \eqref{C:sum-rate_m} is not in a convex form due to the complicated ratio terms as well as the intrinsic coupling between the power allocation. This problem will be solved in the next section.}
\section{Proposed Algorithm for Solving (P3): General case and special cases}\label{SCandGC}
{In this section, we first propose an efficient algorithm to solve the problem for the general case. Then, we consider special cases of problem (P3) to shed interesting insights into the optimal joint beam scheduling and power allocation design.}

\subsection{General Case}\label{General}
First, consider the general case with multiple EH receivers and mulitple ID receivers. Since the sum-rate constraint in \eqref{C:sum-rate_m} is not in a convex form, 
the conventional optimization methods
cannot be directly applied. Thus, an efficient algorithm is proposed to obtain a subopitmal solution to problem (P3). Specifically, we first introduce slack variables $\{S_{m}\}$ and $\{I_{m}\}$ such that 
\begin{align}
	\frac{1}{S_{m}}&=(\mathbf{c}^{\rm ID}_{m})^{T}\mathbf{y},~~\forall m \in \mathcal{M},\\
	I_{m}&=(	\mathbf{c}^{\rm ID}_{m})^{T}\mathbf{\bar{\Lambda}}\mathbf{y}+\sigma^2_{m},~~\forall m \in \mathcal{M}.
\end{align}
Then, problem (P3) can be equivalently transformed into the following form
\begin{subequations}
	\begin{align}
		({\bf P4}):~~\max_{\substack{\mathbf{y},\{S_{m}\},\{I_{m}\}} }  &~~~~	(\mathbf{c}^{\rm EH})^{T}\mathbf{\bar{\Lambda}}\mathbf{y}\nn\\
		\text{s.t.}~~~
		&~~~~ \sum_{m=1}^{M}\log_{2}\left( 1+\frac{1}{S_{m}	I_{m}}\right) \ge R,\label{C:71}
		\\
		&~~~~ 	\frac{1}{S_{m}}\leq(\mathbf{c}^{\rm ID}_{m})^{T}\mathbf{y}, ~~~\forall m \in \mathcal{M},\label{C:xx2} \\
		&~~~~ 	I_{m}\geq(	\mathbf{c}^{\rm ID}_{m})^{T}\mathbf{\bar{\Lambda}}\mathbf{y}+\sigma^2_{m}, ~~~\forall m \in \mathcal{M}, \label{C:sum-rate_m1}\\
		&~~~~
		\eqref{C:sum-power1},\eqref{C:nonn2}.\nn
	\end{align}
\end{subequations}
{Note that the remaining challenge for solving problem (P4) is the constraint in \eqref{C:71}. To address this issue, we first present a useful lemma below.}
\begin{lemma}\label{Le:sca}
	\emph{Given $x>0$ and $y>0$, $f(x,y)=\log_{2}(1+\frac{1}{xy})$ is a convex function with respect to $x$ and $y$.}
\end{lemma}

\begin{proof}
	The proof is similar to that of [Lemma 1, \cite{9139273}] and hence is omitted for brevity.
\end{proof}
{Based on Lemma~\ref{Le:sca}, $\log_{2}(1+\frac{1}{S_{m}	I_{m}})$ in \eqref{C:71} is a joint convex function with respect to $S_{m}$ and $I_{m}$. This thus allows us to apply the SCA method to tackle the constraint \eqref{C:71} as below.}
\begin{lemma}
	\emph{By applying the first-order Taylor expansion, $\log_{2}(1+\frac{1}{S_{m}I_{m}})$ can be lower-bounded as 
		\begin{align}
			\log_{2}&\left( 1+\frac{1}{S_{m}	I_{m}}\right) \ge R^{\rm low}_{m}\triangleq\log_{2}\left(1+\frac{1}{\widetilde{S}_{m}\widetilde{I}_{m}}\right)\nn\\
			&-\frac{(S_{m}-\widetilde{S}_{m})(\log_{2}e)}{\widetilde{S}_{m}+\widetilde{S}_{m}^2\widetilde{I}_{m}}-\frac{(I_{m}-\widetilde{I}_{m})(\log_{2}e)}{\widetilde{I}_{m}+\widetilde{I}_{m}^2\widetilde{S}_{m}},
	\end{align}}
\end{lemma}
where $\{\widetilde{S}_{m},\widetilde{I}_{m}\}$ denote any given feasible points. As such, problem (P4) is approximated as the following problem
\begin{subequations}
	\begin{align}
		({\bf P5}):~~\max_{\substack{\mathbf{y},\{S_{m}\},\{I_{m}\}} }  &~~~~	(\mathbf{c}^{\rm EH})^{T}\mathbf{\bar{\Lambda}}\mathbf{y}\nn\\
		\text{s.t.}~~~
		&~~~~ \sum_{m=1}^{M}R^{\rm low}_{m}\ge R,
		\\
		&~~~~ \eqref{C:sum-power1},\eqref{C:nonn2},\eqref{C:xx2},\eqref{C:sum-rate_m1}.	\nn
	\end{align}
\end{subequations}
Problem (P5) now is a convex optimization problem, which thus can be
efficiently solved via standard solvers such as CVX. {Note that the solution obtained from problem (P5) requires an iterative process until the fractional increase of the objective function is below a threshold $\xi>0$, 
	based on which the beam scheduling and power allocation can be easily recovered. 
	{
\begin{remark}[Extension to nonlinear EH model]\label{Nonlinear}
\emph{For the general case under a nonlinear channel model, the beam scheduling and power allocation design can be extended as follows. First, consider the widely used nonlinear EH model proposed in \cite{boshkovska2015practical}. Then the harvested power  $\Phi_{k}$ at the EH receiver $k$ is expressed as
\begin{align}
	&\Phi_{k}=\frac{[\Psi_{k}-\kappa\Omega]}{1-\Omega},~~~~\Omega= \frac{1}{1+\exp(\varrho \varpi)},\label{Eq:1}\\
	&\Psi_{k}=\frac{\kappa}{1+\exp(-\varrho(Q_{k}-\varpi))},\label{Eq:2}
\end{align}
where $\Psi_{k}$ is a logistic function with its input being the received radio
frequency (RF) power $Q_{k}$ in \eqref{Eq:HP}, and the constant $\Omega$ ensures a zero-input/zero-output response. Besides, $\kappa$, $\varpi$, and $\varrho$ are constant parameters that depend on the employed EH circuit.
Then, the corresponding objective function can be obtained by replacing $Q_{k}$ in \eqref{Eq:Or_OJ} with $\Psi_{k}$ in \eqref{Eq:2}, which admits a sum-of-ratios form and thus is difficult to solve. To solve this problem, by invoking the techniques  
 in [Theorem 1, \cite{boshkovska2015practical}], the  non-convex objective function can be equivalently transformed into a subtractive form, based on which the proposed algorithm for the general case can be readily extended to obtain a suboptimal design.}
\end{remark}}

{To shed useful insights into the joint beam scheduling and power allocation design, we consider in the following three special cases and obtain their optimal solutions.}

\subsection{Special case I: ID Receivers Only}\label{feasibility-check}
First, we consider a special case where there are ID receivers only in the considered mixed-field SWIPT system.
As such, problem (P3) reduces to the following feasibility-check problem
\begin{subequations}
	\begin{align}
		({\bf P6}):~~\text{Find} &~~~~	\mathbf{y}^{\rm ID}\nn\\
		\text{s.t.}
		&~~~~\sum_{m=1}^{M}\log_{2}\left( 1+\frac{(	\mathbf{c}^{\rm ID}_{m})^{T}\mathbf{y}^{\rm ID}}{(	\mathbf{c}^{\rm ID}_{m})^{T}\mathbf{\bar{\Lambda}}\mathbf{y}^{\rm ID}+\sigma^2_{m}}\right) \ge R,\label{P41}\\
		&~~~~
		\mathbf{1}^{T}_{(K+M)\times 1}\mathbf{y}^{\rm ID}\le P_{0},\label{P42}\\
		&~~~~	\mathbf{y}^{\rm ID} \succeq  \mathbf{0},\label{P43}
	\end{align}
\end{subequations}
where $\mathbf{y}^{\rm ID}=\left[\mathbf{0}_{1\times K},\tilde{P}^{\rm ID}_{1},\cdots,\tilde{P}^{\rm ID}_{m},\cdots,\tilde{P}^{\rm ID}_{M}\right]^{T}.$
{The feasibility-check problem (P6) can be equivalently transformed into the following problem \cite{7332956}.
	{
\begin{lemma}\label{Le:water}
	\emph{Problem (P6) is feasible if $R^* \ge R$, where $R^*$ is the maximum value of the following problem (P7)
	\begin{subequations}
		\begin{align}
			({\bf P7}):~\max_{\substack{\mathbf{y}^{\rm ID}} }  &~~	\sum_{m=1}^{M}\log_{2}\left(1+\frac{(	\mathbf{c}^{\rm ID}_{m})^{T}\mathbf{y}^{\rm ID}}{(	\mathbf{c}^{\rm ID}_{m})^{T}\mathbf{\bar{\Lambda}}\mathbf{y}^{\rm ID}+\sigma^2_{m}}\right) \label{P5:OJ}\\
			\text{s.t.}
			&~~
			\eqref{P42},\eqref{P43}.\nn
		\end{align}
\end{subequations}}
\end{lemma}}
\begin{proof}
	It can be shown that problem (P6) is feasible when constraints \eqref{P41}--\eqref{P43} are all satisfied. If the overall power budget is not fully utilized, then the sum-rate can be improved by allocating more power to the ID receivers, thus constraint \eqref{P42} holds with equality. Then, constraint \eqref{P41} (i.e., the objective function of problem (P7)) holds if the maximum value of problem (P7) is larger than $R$, thus completing
	the proof.
\end{proof}}
Problem (P$7$) is a non-convex optimization problem due to the complicated ratio terms in the objective function \eqref{P5:OJ}. To address this issue, the fractional programming (FP) method \cite{8314727} is invoked for transforming the objective function into a more tractable form.
Specifically, as problem (P7) is a concave-convex FP problem,  the quadratic transform \cite{8314727} can be employed to convert the concave-convex objective function into a sequence of concave subproblems, which can be easily solved. To this end, we introduce a set of auxiliary variables $\boldsymbol{\gamma}=[\gamma_{1},\cdots,\gamma_{M}]^T\in\mathbb{C}^{M\times1}$, based on which the objective function in \eqref{P5:OJ} is rewritten as
\begin{align}\label{qdobj}
\begin{aligned}
	f^{\rm obj}(\mathbf{y}^{\rm ID}, \boldsymbol{\gamma})=\sum_{m=1}^{M}\log&\bigg(1+2\gamma_{m}\sqrt{(\mathbf{c}^{\rm ID}_{m})^{T}\mathbf{y}^{\rm ID}}\\
	&-\gamma_{m}^2\left((\mathbf{c}^{\rm ID}_{m})^{T}\mathbf{\bar{\Lambda}}\mathbf{y}^{\rm ID}+\sigma^2_{m}\right)\bigg).
\end{aligned}
\end{align}
Then, we adopt an iterative algorithm to solve problem (P7) by alternately optimizing $\mathbf{y}^{\rm ID}$ and $\gamma_{m}$. Specifically, given any $\mathbf{y}^{\rm ID}$, the optimal $\gamma_{m}$ can be obtained by setting $\partial	f^{\rm obj}/\partial \gamma_{m}$ to be zero, leading to
\begin{equation}\label{Eq:gamma}
	\gamma^{*}_{m}=\frac{(	\mathbf{c}^{\rm ID}_{m})^{T}\mathbf{y}^{\rm ID}}{(	\mathbf{c}^{\rm ID}_{m})^{T}\mathbf{\bar{\Lambda}}\mathbf{y}^{\rm ID}+\sigma^2_{m}},~~~  m\in\mathcal{M}.
\end{equation}
{On the other hand, given any $\gamma_{m}$, finding the optimal $\mathbf{y}^{\rm ID}$ for \eqref{qdobj} is a convex problem, whose optimal solution is well-known to exhibit a \emph{water-filling} structure \cite{8995606}}.

\subsection{Special case II: EH Receivers Only}

Next, we consider the case where EH receivers exist in the system only, for which the optimization variables $\mathbf{y}$ can be expressed as 
$\mathbf{y}^{\rm EH}=\left[\tilde{P}^{\rm EH}_{1},\cdots,\tilde{P}^{\rm EH}_{k},\cdots,\tilde{P}^{\rm EH}_{K},\mathbf{0}_{1\times M}\right]^{T}.$
In this case, our objective is to maximize the weighted sum-power for the near-field EH receivers given the transmit power constraint. Accordingly, problem (P3) is reformulated as
\begin{subequations}
	\begin{align}
		({\bf P8}):~~\max_{\substack{\mathbf{y}^{\rm EH}} }  &~~~~	(\mathbf{c}^{\rm EH})^{T}\mathbf{\bar{\Lambda}}\mathbf{y}^{\rm EH}\label{Eq:OJJ}\\
		\text{s.t.}
		&~~~~
\mathbf{1}^{T}_{(K+M)\times 1}\mathbf{y}^{\rm EH}\le P_{0},\label{Eq:6rate}\\
&~~~~	\mathbf{y}^{\rm EH} \succeq  \mathbf{0}.\label{Eq:6non}
	\end{align}
\end{subequations}
Problem (P8) is now a linear programming (LP) problem. As such, we define an EH priority function and obtain the optimal beam scheduling and power allocation. 
{
\begin{definition}\label{Def2}
	\emph{Define $\rho_{k}\triangleq[(\mathbf{c}^{\rm EH})^{T}\mathbf{\bar{\Lambda}}]_k$ as the EH priority for each EH receiver $k$, in which the one with the highest EH priority is given by
	\begin{equation}
		\rho = \arg\max_{k\in\mathcal{K}}~\rho_{k}= \arg\max_{k\in\mathcal{K}} ~\left[(\mathbf{c}^{\rm EH})^{T}\mathbf{\bar{\Lambda}}\right]_k.
\end{equation}}
\end{definition}}
{
\noindent Note that the defined EH priority corresponds to the coefficient of the objective function in \eqref{Eq:OJJ} for each EH receiver $k$, which essentially represents the power-harvesting ability of EH receiver $k$ in the near-field WPT system.  Note that in the near-field WPT, each EH receiver can not only harvest power from the energy beam steered towards itself, but also the power leakage from other EH receivers.}

\begin{proposition}\label{Le:EUs}
	\emph{In the optimal solution to problem (P8), all transmit power should be allocated to the EH receiver with the highest EH priority. In other words, the optimal power allocation is 
	\begin{equation}\label{os_EH}
			\begin{cases}
			\tilde{P}^{\rm EH}_{\rho} =  P_{0} &  \\
			\tilde{P}^{\rm EH}_{k} = 0  & ~~\forall ~k\in \mathcal{K}/\{\rho\}.
		\end{cases}
	\end{equation}}
\end{proposition}
\begin{proof}
	See Appendix~\ref{App1}.
\end{proof}
{Proposition~\ref{Le:EUs} is intuitively expected since each coefficient in the objective function \eqref{Eq:OJJ} denotes the amount of harvested power if EH receiver $k$ is scheduled; thus, the total transmit power should be allocated to the one with the highest EH priority.}

To sum up, the above two special cases with either ID receivers or EH receivers show different principles for near-/far-field beam scheduling and power allocation. Specifically, the power allocation for far-field ID receivers shares the {water-filling structure}, while that for near-field  EH receivers is allocating all power to the one with the highest EH priority. However, for the new mixed-field SWIPT system with both EH and ID receivers, it is unknown how to schedule different users and properly allocate the BS transmit power.

\begin{figure*}[t]
	\centering
	\subfigure[System setting.]{\includegraphics[width=5cm]{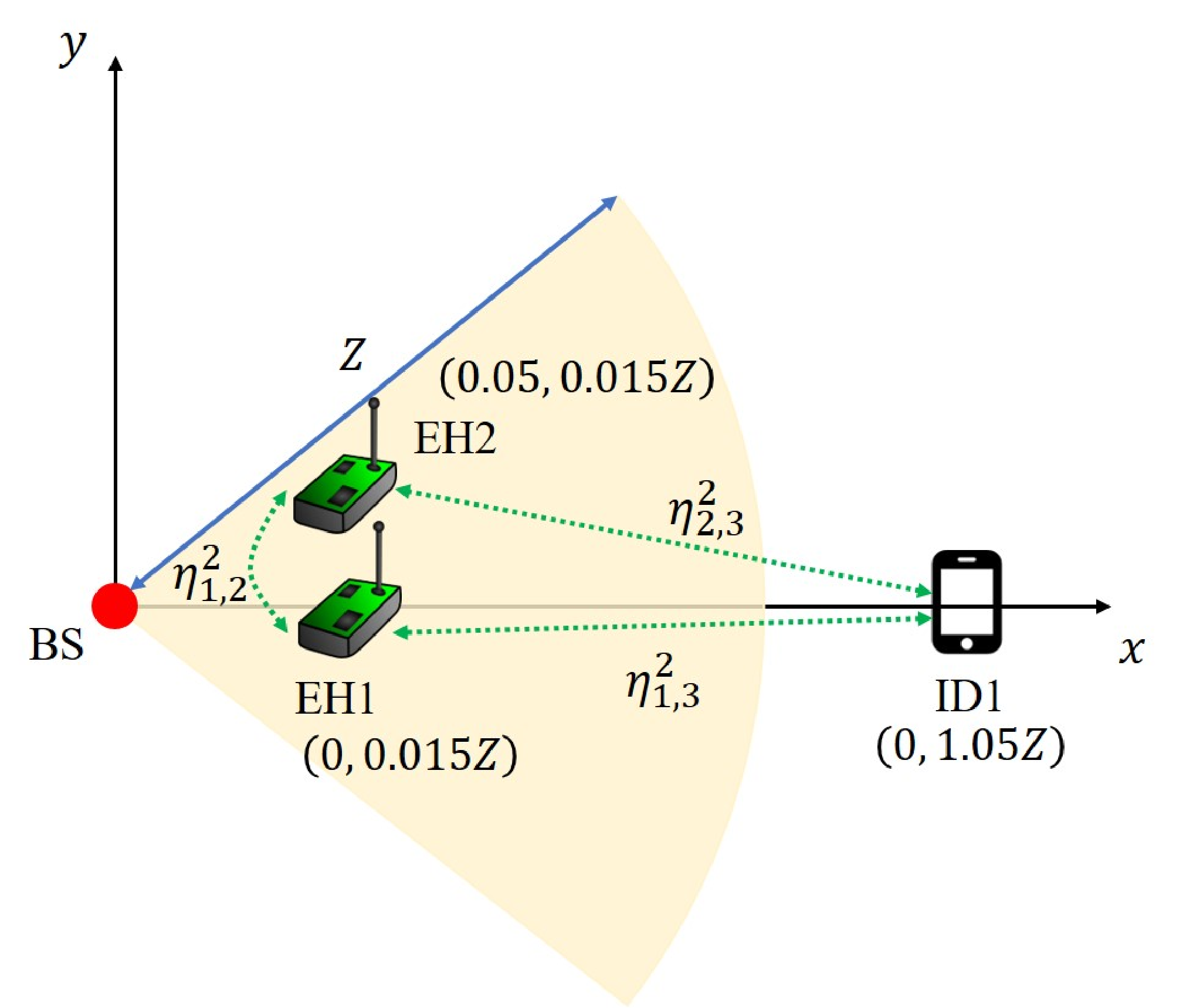}\label{fig:simul_corre}}
	\subfigure[Power allocation vs. correlation.]{\includegraphics[width=5cm]{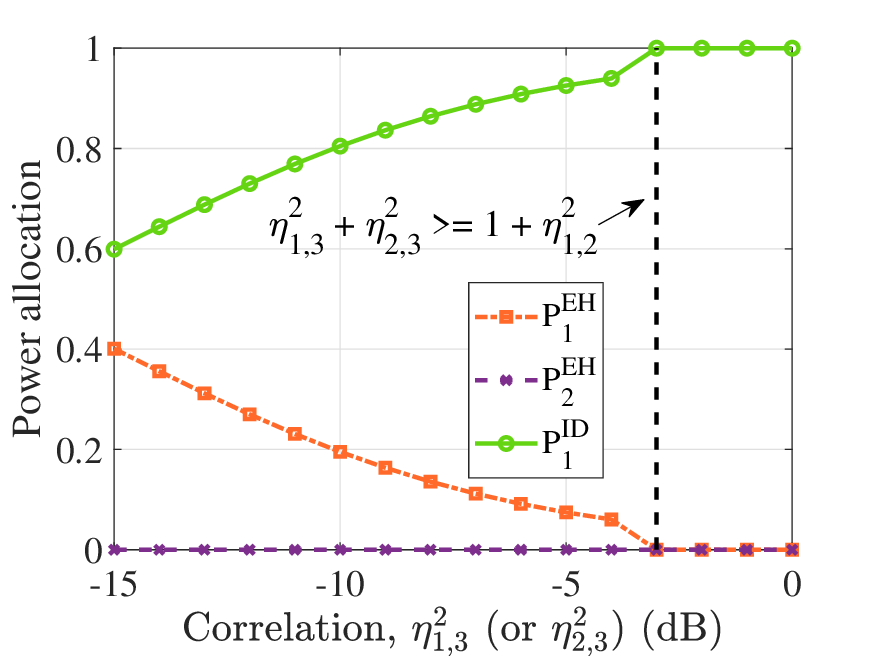}\label{fig:pa_corre}}
	\subfigure[Harvested power vs. correlation.]{\includegraphics[width=5cm]{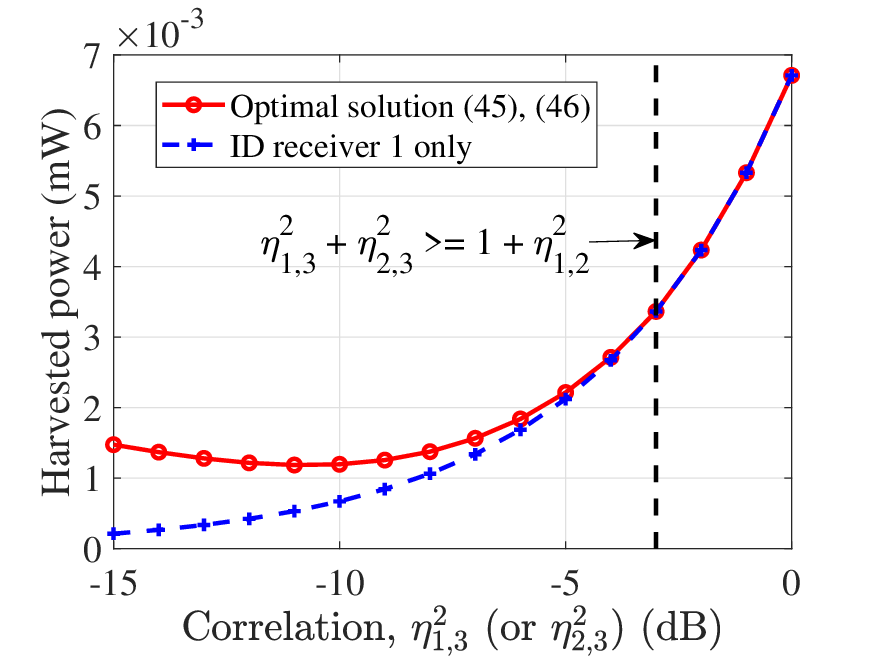}\label{fig:power_corre}}
	\caption{The effect of channel correlation on the optimized power allocation and harvested power with $N=256$, $f=30$ GHz, $P_{0}=30$ dBm, $R=5$, $K=2$ and $M=1$.}\label{fig:cccccc}
\end{figure*}

\subsection{Special case III: Multiple EH receivers and One ID receiver}\label{SC3}
 In order to answer the above question, we consider the following special case with $K$ EH receivers and one ID receiver only.  
 Accordingly, problem (P3) reduces to
 \begin{subequations}
 	\begin{align}
 		({\bf P9}):~\max_{\substack{\mathbf{\hat{y}}} }  &~~	(\mathbf{c}^{\rm EH})^{T}\mathbf{\bar{\Lambda}}\mathbf{\hat{y}}\label{Oj3}\\
 		\text{s.t.}
 		&~~~(2^R-1)((	\mathbf{c}^{\rm ID})^{T}\mathbf{\bar{\Lambda}}\mathbf{\hat{y}}+\sigma^2_{1})\le(\mathbf{c}^{\rm ID})^{T}\mathbf{\hat{y}},\label{C:sum-rate_m3}
 		\\
 		&~~
 		\mathbf{1}^{T}_{(K+1)\times 1}\mathbf{\hat{y}}\le  P_{0},\label{C:sum-power3}\\
 		&~~
 		\mathbf{\hat{y}} \succeq  \mathbf{0}\label{C:nonn3},
 	\end{align}
 \end{subequations}
 where $	\mathbf{\hat{y}}=\left[\tilde{P}^{\rm EH}_{1},\cdots,\tilde{P}^{\rm EH}_{k},\cdots,\tilde{P}^{\rm EH}_{K},\tilde{P}^{\rm ID}_{1}\right]^{T}$ and  $\mathbf{c}^{\rm ID}=\left[\mathbf{0}_{1\times K},g^{\rm ID}_{1}\right]^{T}$.

Problem (P9) is an LP problem, whose optimal solution is given below.
	{
\begin{proposition}\label{Le:manyEUs}
	\emph{The optimal solution to problem (P9) is 
		\begin{itemize}
			\item If $\rho  \neq K+1$, then
				\begin{equation}\label{ONEEU}
				\begin{cases}
					[\mathbf{\hat{y}}]_{\rho}=\frac{ P_{0}-\frac{(2^R-1)\sigma^2_{1}}{[(\mathbf{c}^{\rm ID})]_{K+1}}}{(2^R-1)[\mathbf{\bar{\Lambda}}]_{\rho,K+1}+1},\\
					[\mathbf{\hat{y}}]_{k}=0, ~~~ \forall ~k\in \mathcal{K}/\{\rho\},\\
					[\mathbf{\hat{y}}]_{K+1}=\frac{ P_{0}[\mathbf{\bar{\Lambda}}]_{\rho,K+1}(2^R-1)-\frac{(2^R-1)\sigma^2_{1}}{[(\mathbf{c}^{\rm ID})]_{K+1}}}{(2^R-1)[\mathbf{\bar{\Lambda}}]_{\rho,K+1}+1}.\\
				\end{cases}
				\end{equation}
			\item If $\rho = K+1$, then
			\begin{equation}\label{ALLIU}
				\begin{cases}
					[\mathbf{\hat{y}}]_{k}=0, ~~~\forall ~k\in \mathcal{K},\\
					[\mathbf{\hat{y}}]_{K+1}= P_{0}.\\
				\end{cases}
			\end{equation}
		\end{itemize}
}
\end{proposition}}
\begin{proof}
See Appendix~\ref{App2}.
\end{proof}

	{Proposition \ref{Le:manyEUs} shows an interesting result that the optimal power allocation is determined by the defined EH priority function. Specifically, if the receiver with the highest EH priority is an EH receiver, simply called the best EH receiver (i.e., $\rho \neq K+1$), the optimal power allocation is allocating a part of power to the ID receiver for satisfying the sum-rate constraint, while the remaining power should be allocated to the best EH receiver. On the other hand, if the ID receiver has the highest EH priority, 
	 all power should be allocated to the ID receiver. However, the latter case (i.e., $\arg\max_{k\in\mathcal{K}\cap\{K+1\}} [(\mathbf{c}^{\rm EH})^{T}\mathbf{\bar{\Lambda}}]_k=K+1$) happens only when the correlation coefficients (i.e., $\eta_{K+m,k}$ or $\eta_{k,K+m}$ in $\mathbf{\bar{\Lambda}}$) among them are sufficiently large, such that the ID receiver has the highest EH priority. 
	 For better illustration, we plot respectively in Figs. \ref{fig:pa_corre} and \ref{fig:power_corre} the power allocation and harvested power versus the channel correlation $\eta^2_{1,3}$ (or $\eta^2_{2,3}$). One can observe that, as the channel correlation increases, the harvested-power gap between scheduling the ID receiver only and the optimal solution in \eqref{ONEEU} and \eqref{ALLIU} tends to decrease, and eventually reaches to zero when $\eta^2_{1,3}+\eta^2_{2,3}\ge 1+\eta^2_{1,2}$ (i.e., the ID receiver has the maximum coefficient). In this regard, as shown in Fig. \ref{fig:pa_corre}, allocating all power to the ID receiver is optimal, which is consistent with Proposition \ref{Le:manyEUs}.}

\begin{remark}[How to schedule EH and ID receivers in the mixed-field channels?]\label{Schedulinginmixedfield}
	\emph{Proposition~\ref{Le:manyEUs} shows that in most cases, the optimal policy for the mixed-field SWIPT system with only one ID receiver is allocating a portion of power to the ID receiver for satisfying the rate constraint, while the remaining power is allocated to the EH receiver with the highest EH priority. Note that this result is different from the conventional far-field SWIPT system, for which the optimal policy is allocating all power to the ID receivers \cite{6860253,8941080}, as illustrated in Fig. \ref{fig:howschedule}. The reasons are as follows. First, unlike the far-field beamforming that steers the beam along a specific direction, near-field beamforming focuses most of its energy on a specific region (or range), termed as \emph{energy focusing} in the existing literature \cite{9738442,9957130}. Thus, the EH efficiency in the near-field region is significantly improved, which encourages to allocate more power to the EH receivers for maximizing their harvested sum-power. Second, the energy focusing effect also indicates that less EH power is leaked to the ID receiver. Combining this effect with the severe path-loss between the BS and ID receiver, it is expected that the far-field ID receiver will receive less interference from the EH signals, as compared to conventional far-field SWIPT system.
		On the other hand, the transmit power to the ID receiver can charge the near-field EH receivers when they are located near the ID receiver in the spatial domain (as discussed in Remark \ref{Fresnel}).
		}
\end{remark}

\begin{figure}[t]
	\centering
	\includegraphics[width= 0.49\textwidth]{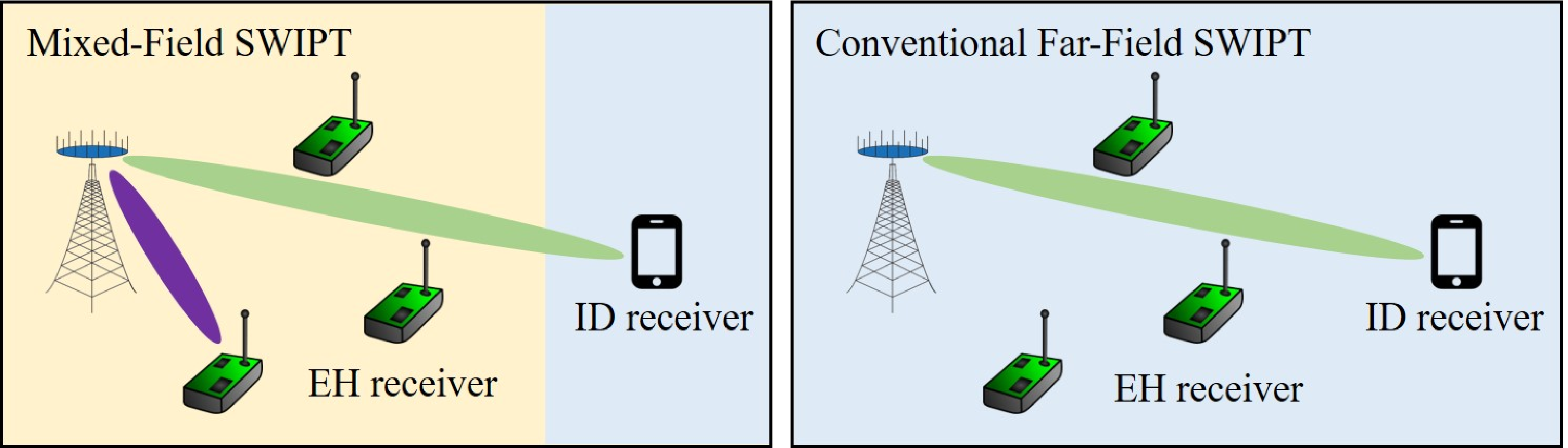}
	\caption{Illustration of the scheduling in mixed-field channels and conventional far-field channel, respectively.}\label{fig:howschedule}
\end{figure}
%

{
\begin{remark}[Convergence and complexity analysis]
	\emph{First, consider the convergence of the proposed SCA-based algorithm. As problem (P5) is solved optimally in each iteration and its objective value serves as a lower bound on that of problem (P4), it can be shown that the objective value of problem (P4) is non-decreasing. Moreover, since the system sum-power is upper-bounded by a finite value, the proposed algorithm is guaranteed to converge. Next, the overall complexity of the proposed SCA-based algorithm is determined by the number of iterations, denoted by $I_{\rm iter}$, and the number of the optimization variables (i.e., $K+M$). For each iteration, the computational complexity for solving problem (P5) by the interior-point method can be characterized as $\mathcal{O}((K+M)^{3.5})$ \cite{ding2023joint}. As a result, the total computational complexity of the proposed SCA-based algorithm is $\mathcal{O}(I_{\rm iter}(K+M)^{3.5})$. }
\end{remark}}

\section{Numerical Results}\label{SE:NR}

\begin{figure}[t]
	\centering
	\includegraphics[width=2.7in]{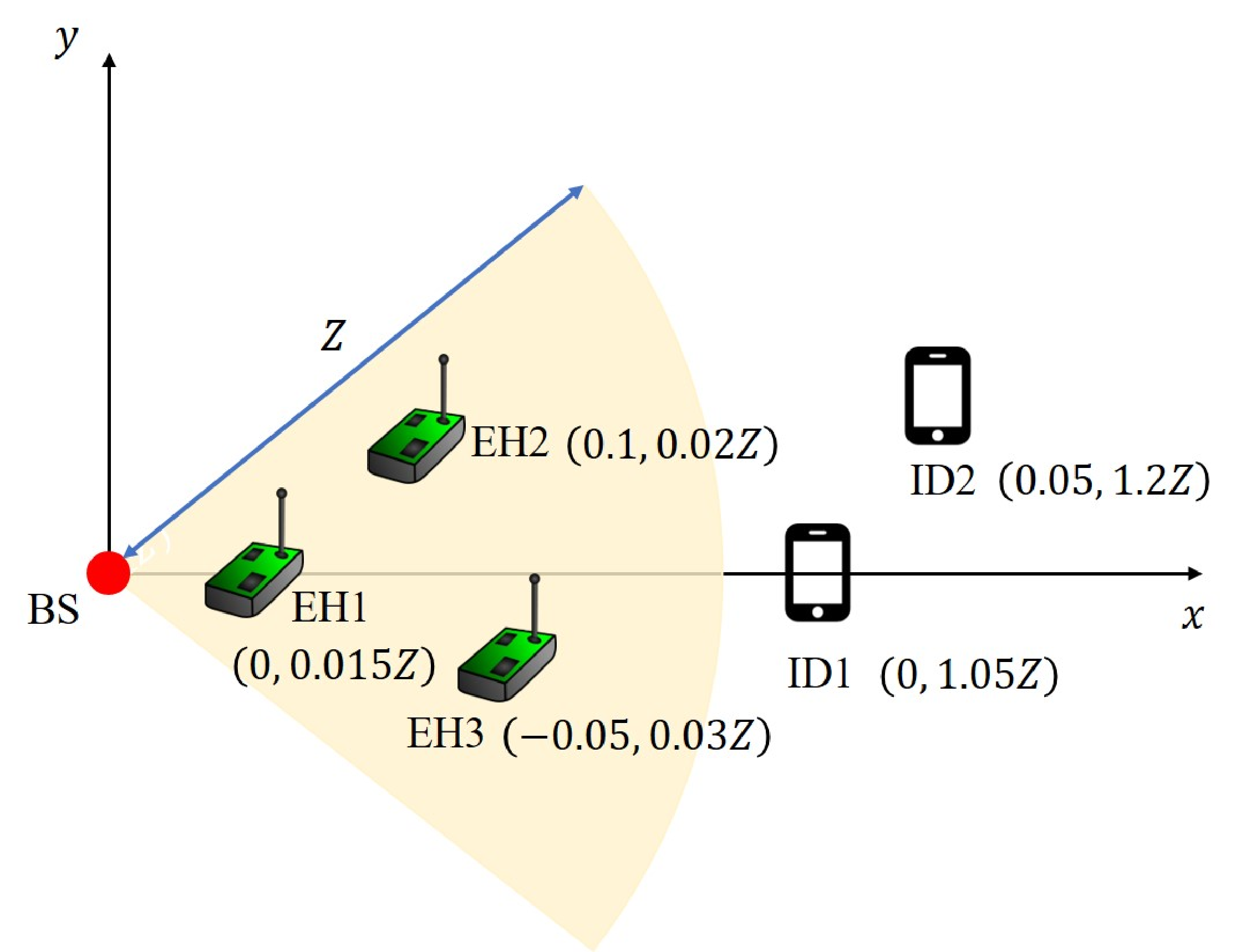}
	\caption{ Simulation setup for a mixed-field SWIPT system. }
	\label{fig:simlu}
\end{figure}

In this section, we present numerical results to show the effectiveness of the proposed scheme.
The considered system model is illustrated in Fig.~\ref{fig:simlu}, where one BS with $N=256$ antennas serves three EH receivers and two ID receivers. Specifically, under the polar coordinate, the EH receivers are located at $(0,0.015Z)$, $(0.1,0.02Z)$ and $(-0.05,0.03Z)$, while the two ID receivers are located  at $(0,1.05Z)$ and $(0.05,1.2Z)$. Unless specified otherwise, the system parameters adopted in our numerical results are listed in Table~\ref{Table1}.
\begin{table}[t]
	\centering
	\caption{Simulation Parameters}
	\label{Table1}
		\begin{tabular}{|c|c|}
			\hline
			{\bf{Parameter}} & {\bf{Value}}  \\
			\hline
			Number of BS antennas & $N=256$\\\hline
				Number of EH receivers & $K=3$\\\hline
					Number of ID receivers & $M=2$\\\hline
			 Carrier frequency & $f=30$ GHz \\\hline
			Reference path-loss & $\beta=(\lambda/4\pi)^2=-62$ dB \\ \hline
			Maximum transmit power  & $P_0 =30$ dBm\\ \hline
			Noise power & $\sigma_{m}^2=-80$ dBm, $m\in\mathcal{M}$\\ \hline
			Energy harvest efficiency & $\zeta=50\% $ \cite{6860253}\\
			\hline	
			Antenna spacing & $d=\frac{\lambda}{2}=0.005$ m \\ \hline
			Power weight & $\alpha_{k}=1, \forall k \in \mathcal{K}$ \cite{8941080,9110849}\\ 
			\hline	
			Convergence threshold & $\xi=0.001$ \\ \hline
			Number of NLoS paths & $L_{k}=L_{m}=2, \forall k,m$\\ 
			\hline
		\end{tabular}
	\end{table}
\begin{figure}[t]
	\centering
	\includegraphics[width=3in]{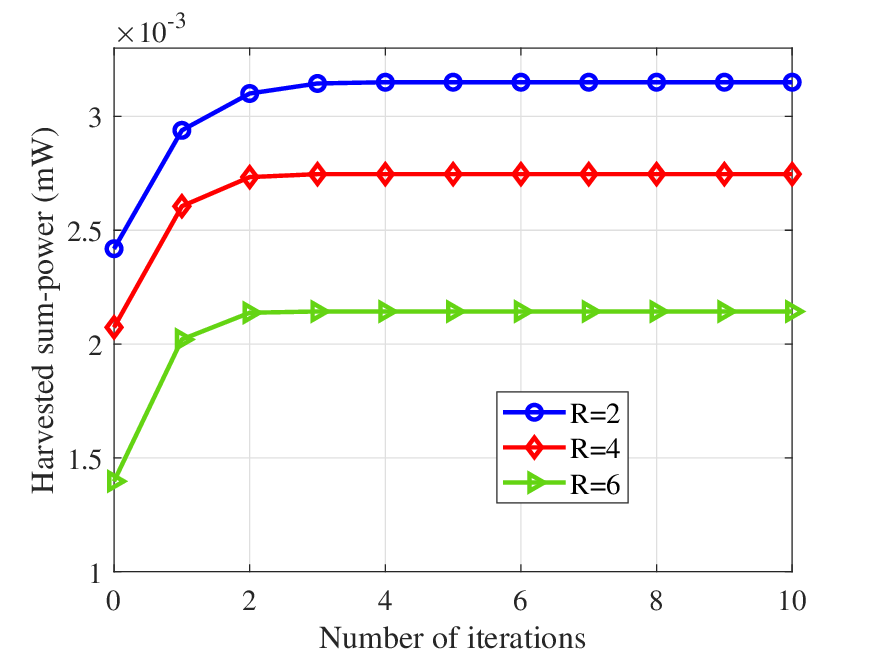}
	\caption{Convergence of the proposed scheme for different sum-rate requirements.}
	\label{fig:conver}
\end{figure}

Moreover, for performance comparison, we consider the following four benchmark schemes.
\begin{itemize}
	\item 	\emph{Exhaustive-search scheme:} This scheme enumerates  all possible beam scheduling combinations with optimized power allocation, and then selects the best one that achieves the maximum harvested sum-power.
	\item 	{\emph{Far-field SWIPT scheme:} Similar to the conventional far-field SWIPT beamforming and scheduling \cite{6860253}, this scheme focuses on the optimization of beam scheduling and power allocation for ID receivers only, while assuming no energy beams steered towards the EH receivers.}
	\item	\emph{Greedy scheduling + optimized power allocation (GS+OPA) scheme}, for which the EH receiver with the highest EH priority and the ID receiver with the best channel condition are scheduled with optimized power allocation.
	\item	\emph{Optimal scheduling + equal power allocation (OS+EPA) scheme}, which adopts the optimal beam scheduling and allocates equal power allocation for all selected receivers.
	\item	\emph{All scheduling + equal power allocation (AS+EPA) scheme}, for which all EH and ID receivers are scheduled with equal power allocation.
\end{itemize}

\subsection{Convergence of SCA-based Scheme}
To show the convergence of the proposed SCA-based scheme, we plot in Fig.~\ref{fig:conver} the harvested sum-power versus the number of iterations for different sum-rate requirements, $R$. It is observed that the proposed SCA-based scheme always converges after a few iterations (e.g., $5$ iterations) under different sum-rate requirements.
\begin{figure*}[t]
	\centering
	\subfigure[Effect of BS transmit power with $K=3$, $M=2$, $R=5$.]{\includegraphics[width=3.1in]{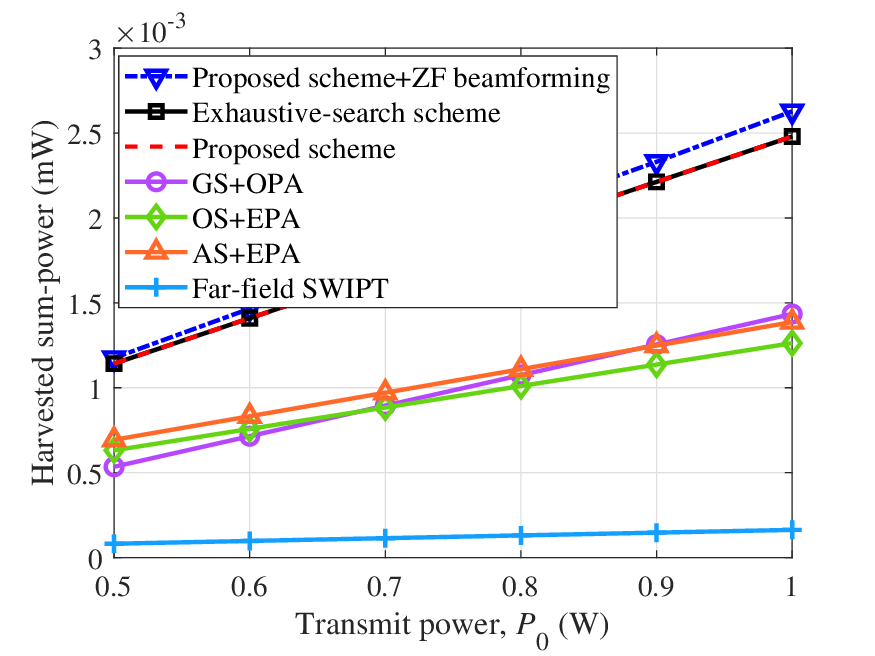}\label{fig:transpow}}
	\subfigure[Effect of sum-rate requirement with $K=3$, $M=2$.]{\includegraphics[width=3.1in]{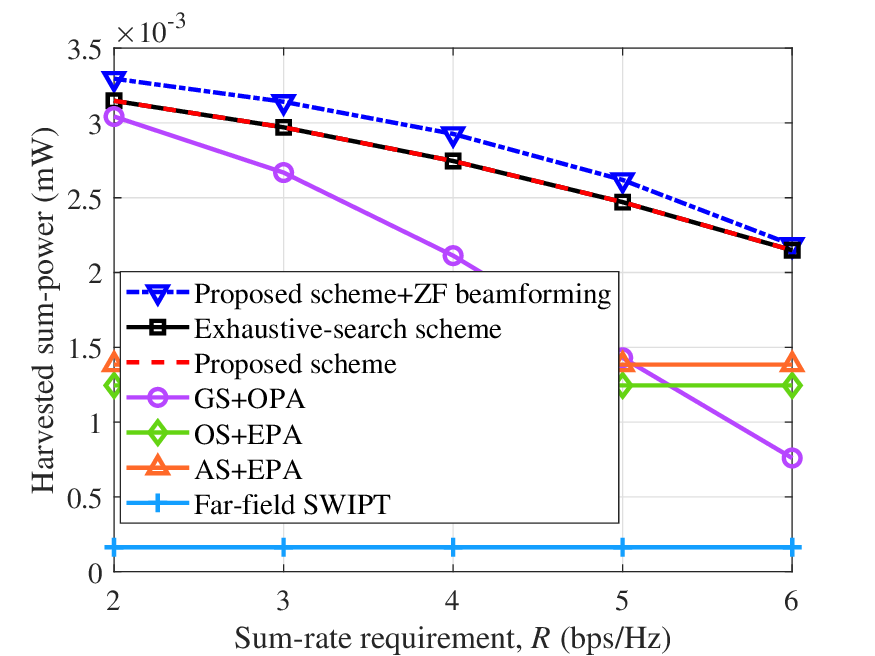}\label{fig:ratecons}}
	\subfigure[Effect of number of EH receivers with $M=2$, $R=5$.]{\includegraphics[width=3.1in]{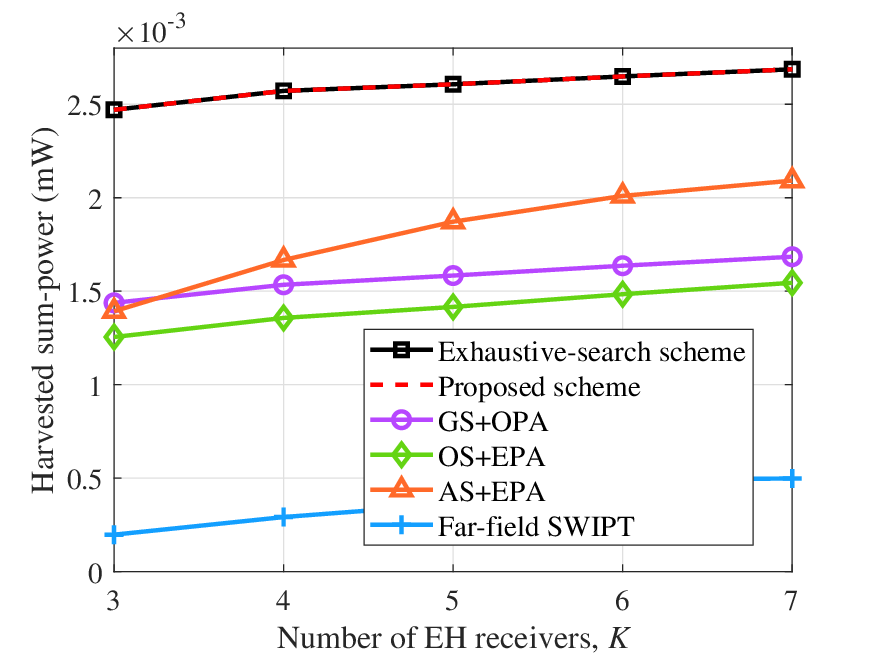}\label{fig:numEUs}}	
	\subfigure[Effect of number of ID receivers with $K=3$, $R=5$.]{\includegraphics[width=3.1in]{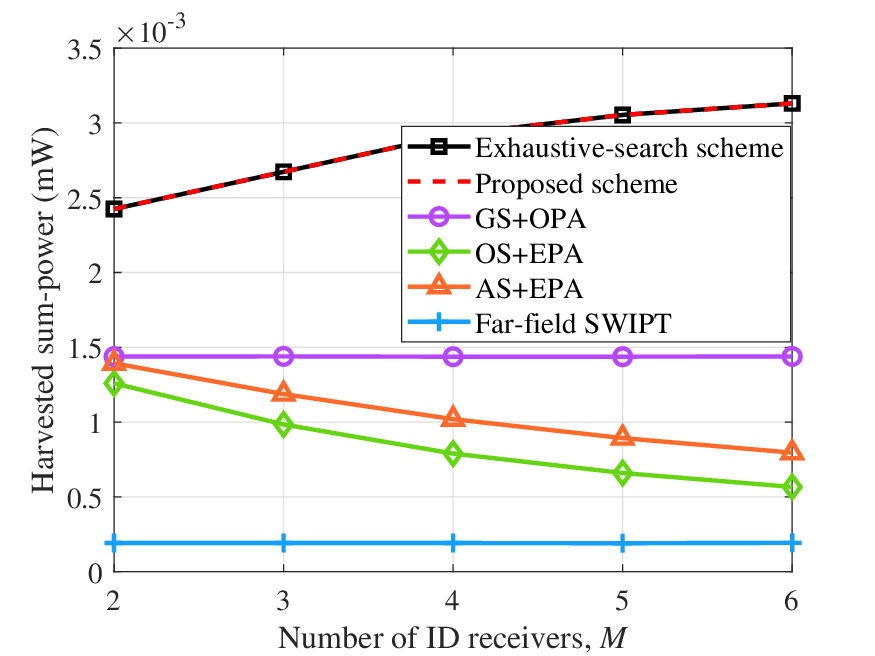}\label{fig:numIUs}}
	\caption{Harvested sum-power versus system parameters.}
\end{figure*}
\subsection{Effect of Maximum Transmit Power}\label{DBBB}
To show the effect of the maximum transmit power $P_0$, we plot in Fig. \ref{fig:transpow} the harvested sum-powers by different schemes versus $P_0$. First, it is observed that the harvested sum-powers by all schemes monotonically increases with the transmit power. {Second, the proposed scheme achieves close performance with that of the exhaustive-search scheme, while it only suffers a small performance loss compared to the proposed scheme combined with the ZF digital beamforming, especially when $P_0$ is small.} Besides, the proposed scheme achieves significant performance gain as compared to other benchmark schemes e.g., around twice of the harvested sum-power by the OS+EPA scheme, which demonstrates the importance of the beam scheduling and power allocation design. Last, one can observe that the harvested sum-power of the proposed scheme increases much faster with the BS transmit power than those of the benchmark schemes, {\color{black}thanks to the effectiveness of the designed beam scheduling and power allocation.}

\subsection{Effect of Sum-Rate Requirement}
In Fig.~\ref{fig:ratecons}, we compare the harvested sum-powers by different schemes versus the sum-rate requirement, $R$. First, for the proposed scheme and the GS+OPA scheme, it is observed that the harvested sum-power monotonically decreases as the sum-rate requirement increases. This is because when the sum-rate requirement increases, more power should be allocated to the ID receivers for satisfying the sum-rate requirement, hence leaving less power for the EH receivers. One interesting observation is that the gap between the proposed scheme and the GS+OPA scheme becomes larger with the growing sum-rate requirement. {This is because for the proposed scheme, the power allocated to different ID receivers follows the ``water-filling" structure, while the GS+OPA  scheme allocates transmit power to one ID receiver only, thus resulting in a lower EH efficiency.} Moreover, one can observe that the AS+EPA scheme is superior to the OS+EPA scheme. The reason is that the overall transmit power allocated to the EH receivers in the AS+EPA scheme is higher than that of the OS+EPA scheme, thus leading to a higher harvested sum-power performance. {Besides, it is observed that the far-field SWIPT scheme keeps unchanged with the sum-rate requirement. This is because in the mixed-field SWIPT, the far-field SWIPT scheme still admits a same solution structure with that in \cite{6860253}, for which only one ID receiver is scheduled.  }
\subsection{Effect of Number of EH Receivers}
In Fig. \ref{fig:numEUs}, we evaluate the performance of the proposed scheme under different number of EH receivers, $K$. Apart from the three EH receivers in Fig.~\ref{fig:simlu}, the newly added EH receivers are uniformly distributed in an area with a radius range 
$\left[ 0.015Z,0.3Z\right] $ and a spatial angle range $\left[ - \frac{\pi}{3},\frac{\pi}{3}\right]$. We plot in Fig. \ref{fig:numEUs} the harvested sum-power by different schemes with  $M=2$, $R=5$. First, it is observed that the proposed scheme significantly outperforms the benchmark schemes under different $K$. Second, we observe that the harvested sum-power by all schemes increase with $K$. {This is because the scheduled EH and ID beams can charge the newly added EH receiver in the considered SWIPT system, as shown in Example \ref{Ex:PL}. Besides, it is worthy noting that for the OS+EPA scheme, more power is allocated to the EH receivers when there are more EH receivers, leading to the highest increase rate of the harvested sum-power.}

\subsection{Effect of Number of ID Receivers}
Next, we show the effect of the number of ID receivers on the harvested sum-power by different schemes in Fig. \ref{fig:numIUs}. Apart from the original two ID receivers, the newly added ID receivers are uniformly distributed in an area with a radius range 
$\left[ 1.05Z,1.3Z\right] $ and a spatial angle range $\left[ - \frac{\pi}{3},\frac{\pi}{3}\right] $. First, it is observed that the harvested sum-power with the proposed scheme monotonically increases with $M$. This is intuitively expected since when $M$ increases, the scheduled ID beams steered towards the newly added ID receivers can also charge the selected EH receiver. Besides, the harvested sum-power by both OS+EPA and AS+EPA schemes appear to decrease with $M$. This is because less transmit power is allocated to the scheduled EH receiver with a larger $M$, thus reducing the total harvested sum-power. {In addition, one interesting observation is that when $M$ increases, the harvested sum-powers of the GS+OPA and far-field SWIPT schemes keep unchanged. This is because the newly added ID receivers with worse channel conditions do not affect the two schemes.}
{
\subsection{Complexity Comparison}
Last, we compare the performance and complexity of all schemes. Among others, the exhaustive-search scheme yields the best performance, while it incurs a prohibitively high complexity scaling in an order of $\mathcal{O}(2^{K+M}I_{\rm iter}(K+M)^{3.5})$. It is worth noting that the proposed scheme achieves very close weighted sum-power performance to the exhaustive-search scheme, but with a much lower computational complexity, i.e., $\mathcal{O}(I_{\rm iter}(K+M)^{3.5})$. The GS+OPA, OS+EPA, and AS+EPA schemes enjoy low complexity, but at the cost of considerable performance loss. Besides, the far-field SWIPT scheme achieves the worst performance, while its computation complexity is comparable to that of our proposed scheme, i.e.  $\mathcal{O}(I_{\rm iter}M^{3.5})$.
}

\section{Conclusions}\label{CL}
In this paper, we considered a new mixed-field SWIPT system comprising both near-field EH receivers and far-field ID receivers. An optimization problem was formulated to jointly optimize the BS beam scheduling and power allocation to maximize the weighted sum-power harvested at all EH receivers subject to the ID sum-rate and BS transmit power constraints. To gain useful insights, we showed that when there are one ID receiver and multiple EH receivers, the optimal design is allocating a portion of power to the ID receiver for satisfying the sum-rate constraint, while the remaining power is allocated to one EH receiver with the highest EH priority. {Moreover, for the general case, we proposed an efficient algorithm to obtain a suboptimal solution by utilizing the binary variable elimination and SCA methods.} Numerical results are presented to corroborate the effectiveness of the proposed scheme against various benchmark schemes. {In the future, it is interesting to study hardware-efficient hybrid beamforming design and new mixed-field multi-access techniques to further improve the XL-array system performance.}

\begin{appendices}
	\section{}\label{App1}
	First, the Lagrange function of problem (P8)
	is given by
	\begin{align}
		\mathcal{L}(\mathbf{y}^{\rm EH},\tau,\boldsymbol{\mu})=-(\mathbf{c}^{\rm EH})^{T}\mathbf{\bar{\Lambda}}\mathbf{y}^{\rm EH}&+\tau(\mathbf{1}^{T}_{(K+M)\times 1}\mathbf{y}^{\rm EH}- P_0)\nn\\
		&-\sum_{k=1}^{K}[\boldsymbol{\mu}]_k[\mathbf{y}^{\rm EH}]_k,
	\end{align}
	where $\tau$ and $\boldsymbol{\mu}=[\mu_1,\mu_2,\cdots,\mu_{K}]^T$ denote the corresponding dual variables for constraints \eqref{Eq:6rate} and \eqref{Eq:6non}, respectively.
	Then, the  Karush-Kuhn-Tucker (KKT) conditions for problem (P8) are given by
	\begin{align}
			&\text{K1}:	\nabla_{\mathbf{y}^{\rm EH}} \mathcal{L}(\mathbf{y}^{\rm EH},\tau,\boldsymbol{\mu})
			=\mathbf{0},\nn\\
			&\text{K2}:	\tau(\mathbf{1}^{T}_{(K+M)\times 1}\mathbf{y}^{\rm EH}- P_0)=0,\nn\\
			&\text{K3}:	\tau\ge0,~~\text{K4}: \boldsymbol{\mu}\succeq \mathbf{0},
			~~\text{K5}:\boldsymbol{\mu}^T\mathbf{y}^{\rm EH}=\mathbf{0},\nn
	\end{align}
	By checking conditions K1 and K3, one can readily obtain that $\tau>0$, and thus we have $\mathbf{1}^{T}_{(K+M)\times 1}\mathbf{y}^{\rm EH}- P_0=0$. Next, by unveiling the structure of K1, we can easily show that the optimal $\tau$ always satisfies $	\tau = \max_{k\in\mathcal{K}} ~\left[(\mathbf{c}^{\rm EH})^{T}\mathbf{\bar{\Lambda}}\right]_k$, i.e., the highest EH priority in Definition \ref{Def2}. Combining this with K1 and K4, we can show that if $[\boldsymbol{\mu}]_\rho=0$, $[\boldsymbol{\mu}]_k > 0, \forall k \in \mathcal{K}/\{\rho\}$, both K1 and K4 hold; while on other hand, if there exists  $[\boldsymbol{\mu}]_k = 0$, K1 does not hold.
%
Then, by checking K5, the optimal solution to problem (P8) is given in \eqref{os_EH}.
Combining the above leads to the proof of
	Proposition~\ref{Le:EUs}. 
	\section{}\label{App2}
		Similar to the proof of Proposition~\ref{Le:EUs}, the Lagrange function of problem (P9)
	is given by
	\begin{align}
		\mathcal{L}(\mathbf{\hat{y}},\tau,\nu,\boldsymbol{\mu})=&-(\mathbf{c}^{\rm EH})^{T}\mathbf{\bar{\Lambda}}\mathbf{\hat{y}}+\tau(\mathbf{1}^{T}_{(K+1)\times 1}\mathbf{\hat{y}}- P_0)\nn\\&+\nu[(2^R-1)((	\mathbf{c}^{\rm ID})^{T}\mathbf{\bar{\Lambda}}\mathbf{\hat{y}}+\sigma^2_{1})-(\mathbf{c}^{\rm ID})^{T}\mathbf{\hat{y}}]\nn\\
		&-\sum_{i=1}^{K+1}[\boldsymbol{\mu}]_i[\mathbf{\hat{y}}]_i,
	\end{align}
	where $\tau$, $\nu$, and $\boldsymbol{\mu}=[\mu_1,\mu_2,\cdots,\mu_{K+1}]^T$ denote corresponding dual variables for constraints \eqref{C:sum-rate_m3}--\eqref{C:nonn3}, respectively.
	Thus, the KKT conditions for problem (P9) are given by
	\begin{align}
			&\text{K1}:	\nabla_{\mathbf{\hat{y}}} \mathcal{L}(\mathbf{\hat{y}},\tau,\nu,\boldsymbol{\mu})
			=\mathbf{0},\nn\\
			&\text{K2}:	\nu[(2^R-1)((	\mathbf{c}^{\rm ID})^{T}\mathbf{\bar{\Lambda}}\mathbf{\hat{y}}+\sigma^2_{1})-(\mathbf{c}^{\rm ID})^{T}\mathbf{\hat{y}}]=0,\nn\\
			&\text{K3}:(2^R-1)((	\mathbf{c}^{\rm ID})^{T}\mathbf{\bar{\Lambda}}\mathbf{\hat{y}}+\sigma^2_{1})-(\mathbf{c}^{\rm ID})^{T}\mathbf{\hat{y}}\le 0, ~~\nn\\
			&\text{K4}:	\tau(\mathbf{1}^{T}_{(K+1)\times 1}\mathbf{\hat{y}}- P_0)=0,~~\text{K5}:\mathbf{1}^{T}_{(K+1)\times 1}\mathbf{\hat{y}}- P_0\le 0, \nn\\
			&\text{K6}:\boldsymbol{\mu}^T\mathbf{\hat{y}}=\mathbf{0}, ~~
			\text{K7}:	\tau\ge0,\nu \ge0,\boldsymbol{\mu}\succeq \mathbf{0},\nn	
	\end{align}
	By checking condition K3, one can readily obtain that $[\mathbf{\hat{y}}]_{K+1} $ must be greater than $0$, and then we have $[\boldsymbol{\mu}]_{K+1}=0$ according to K6. Then, for the dual variable $[\boldsymbol{\mu}]_{K+1}=0$, by using condition K1, it can be easily shown that neither $\tau=0$ nor $\nu=0$. Therefore, based on conditions K2 and K4, we can have the following result
	\begin{align}
		&(2^R-1)((	\mathbf{c}^{\rm ID})^{T}\mathbf{\bar{\Lambda}}\mathbf{\hat{y}}+\sigma^2_{1})-(\mathbf{c}^{\rm ID})^{T}\mathbf{\hat{y}}= 0,\label{KKT1}\\
		&\mathbf{1}^{T}_{(K+1)\times 1}\mathbf{\hat{y}}- P_{0}= 0.\label{KKT2}
	\end{align}
	Next, similar to Definition \ref{Def2}, among multiple EH receivers and one ID receiver, the one with the highest EH priority $\rho$  is given by
		\begin{equation}
		\rho = \arg\max_{k\in\mathcal{K}\cap\{K+1\}} ~\left[(\mathbf{c}^{\rm EH})^{T}\mathbf{\bar{\Lambda}}\right]_k,
	\end{equation}	
where two cases are considered below.
	\begin{itemize}
		\item {\bf (Case 1: ${{\rho} \neq K+1}$)} In this case, one of the EH receives has the highest EH priority. By using condition K1, we have $[\boldsymbol{\mu}]_\rho =0 $ and $ [\boldsymbol{\mu}]_k > 0 , \forall k \neq \rho \neq K+1$, which agrees with the condition K7. Then, according to K6, the following equation holds
		\begin{equation}\label{KKT3}
			[\mathbf{\hat{y}}]_k = 0, ~~~\forall k \neq \rho  \neq K+1.
		\end{equation}
		Combining \eqref{KKT1}, \eqref{KKT2} and \eqref{KKT3} leads to the following optimal solution to problem (P9)
		\begin{equation}
			\begin{cases}
				[\mathbf{\hat{y}}]_{\rho }=\frac{P-\frac{(2^R-1)\sigma^2_{1}}{[(\mathbf{c}^{\rm ID})]_{K+1}}}{(2^R-1)[\mathbf{\bar{\Lambda}}]_{\rho,K+1}+1},\\
				[\mathbf{\hat{y}}]_{k}=0, ~~~\forall k \neq \rho  \neq K+1,\\
				[\mathbf{\hat{y}}]_{K+1}=\frac{P[\mathbf{\bar{\Lambda}}]_{\rho,K+1}(2^R-1)-\frac{(2^R-1)\sigma^2_{1}}{[(\mathbf{c}^{\rm ID})]_{K+1}}}{(2^R-1)[\mathbf{\bar{\Lambda}}]_{\rho,K+1}+1}.\\
			\end{cases}
		\end{equation}
		\item {\bf (Case 2: ${{\rho} = K+1}$)} In this case, the ID receiver has the highest EH priority. Then, for any two dual variables $[\boldsymbol{\mu}]_i$ and $ [\boldsymbol{\mu}]_j$, $i \neq j \neq K+1$, by checking K1, we have $[\boldsymbol{\mu}]_i>0$ and $ [\boldsymbol{\mu}]_j>0$. Thus, condition K6 holds if the following conditions are satisfied
		\begin{equation}\label{KKT4}
			[\mathbf{\hat{y}}]_k = 0, ~~~\forall k  \neq K+1.
		\end{equation}
		In this regard, combining \eqref{KKT1}, \eqref{KKT2} and \eqref{KKT4} leads to the following optimal solution to problem (P9)
		\begin{equation}
			\begin{cases}
				[\mathbf{\hat{y}}]_{k}=0, ~~~\forall k \neq K+1,\\
				[\mathbf{\hat{y}}]_{K+1}=P_{0}.\\
			\end{cases}
		\end{equation}
	\end{itemize}
	Combining the above leads to the result in
	Proposition~\ref{Le:manyEUs}. 
\end{appendices}
\bibliographystyle{IEEEtran}
\bibliography{IEEEabrv,BS_Ref}
\end{document}